%% file: main.tex
\pdfoutput=1
\documentclass[a4paper,UKenglish,cleveref, autoref, thm-restate]{lipics-v2021}
\hideLIPIcs  %uncomment to remove references to LIPIcs series (logo, DOI, ...), e.g. when preparing a pre-final version to be uploaded to arXiv or another public repository

\title{Counterfactual Explanations for MITL Violations} %TODO Please add

\author{Bernd Finkbeiner}{CISPA Helmholtz Center for Information Security, Saarbrücken, Germany}{finkbeiner@cispa.de}{0000-0002-4280-8441}{}%TODO mandatory, please use full name; only 1 author per \author macro; first two parameters are mandatory, other parameters can be empty. Please provide at least the name of the affiliation and the country. The full address is optional. Use additional curly braces to indicate the correct name splitting when the last name consists of multiple name parts.

\author{Felix Jahn}{Saarland University, Saarland Informatics Campus, Saarbrücken, Germany}{felix.jahn@uni-saarland.de}{0000-0003-4851-3385}{}

\author{Julian Siber}{CISPA Helmholtz Center for Information Security, Saarbrücken, Germany}{julian.siber@cispa.de}{0000-0003-0842-0029}{}

\authorrunning{B.\ Finkbeiner, F.\ Jahn, and J.\ Siber} %TODO mandatory. First: Use abbreviated first/middle names. Second (only in severe cases): Use first author plus 'et al.'

\Copyright{Bernd Finkbeiner, Felix Jahn, and Julian Siber} %TODO mandatory, please use full first names. LIPIcs license is "CC-BY";  http://creativecommons.org/licenses/by/3.0/


 %TODO mandatory; please add comma-separated list of keywords

\supplement{}
\supplementdetails[subcategory={}, cite={}, swhid={swh:1:snp:a9716ce4356b5dd9f88d9df49a2f8e2eb84cb8b1}]{Software}{ https://github.com/reactive-systems/rt-causality}
\funding{This work was partially supported by the DFG in project 389792660 (Center for Perspicuous Systems, TRR 248) and by the ERC Grant HYPER (No. 101055412).}%optional, to capture a funding statement, which applies to all authors. Please enter author specific funding statements as fifth argument of the \author macro.

\nolinenumbers %uncomment to disable line numbering

\EventEditors{John Q. Open and Joan R. Access}
\EventNoEds{2}
\EventLongTitle{44th IARCS Annual Conference on
Foundations of Software Technology and Theoretical Computer Science (FSTTCS 2024)}
\EventShortTitle{FSTTCS 2024}
\EventAcronym{FSTTCS}
\EventYear{2024}
\EventDate{December 24--27, 2016}
\EventLocation{Little Whinging, United Kingdom}
\EventLogo{}
\SeriesVolume{42}
\ArticleNo{23}
\usepackage{comment}

\usepackage{graphicx}
\usepackage{amsmath}
\usepackage{amssymb}
\usepackage{amsfonts}
\usepackage{comment}
\usepackage{mathtools}
\usepackage{hyperref}
\usepackage{tikz}
\usetikzlibrary{arrows,automata,shapes.multipart,positioning}
\usepackage {mathpartir}
\usepackage[ruled,linesnumbered]{algorithm2e} 
\usepackage{hhline}
\usepackage{ltl}
\usepackage{tabularray}
\usepackage{makecell}
\UseTblrLibrary{booktabs}

\begin{document}
\include{definitions_verification}

\maketitle

\begin{abstract}
MITL is a temporal logic that facilitates the verification of real-time systems by expressing the critical timing constraints placed on these systems. MITL specifications can be checked against system models expressed as networks of timed automata. A violation of an MITL specification is then witnessed by a \emph{timed trace} of the network, i.e., an execution consisting of both discrete actions and real-valued delays between these actions. Finding and fixing the root cause of such a violation requires significant manual effort since both discrete actions and real-time delays have to be considered. In this paper, we present an automatic explanation method that eases this process by computing the root causes for the violation of an MITL specification on the execution of a network of timed automata. This method is based on newly developed definitions of counterfactual causality tailored to networks of timed automata in the style of Halpern and Pearl's actual causality. We present and evaluate a prototype implementation that demonstrates the efficacy of our method on several benchmarks from the literature. 
\end{abstract}

\section{Introduction}

Networks of timed automata are a popular formalism to model a wide range of real-time systems such as automotive controllers~\cite{GerkeEFP10,LindahlPY98} and communication protocols~\cite{DavidY00,Havelund97}. These models can be automatically checked against specifications in \emph{Metric Interval Temporal Logic}~(MITL)~\cite{alur1996}, a real-time extension of linear-time temporal logic that allows to constrain temporal operators to non-singleton intervals over the real numbers. In case a network of timed automata does not satisfy an MITL specification, a model-checking procedure will return an execution of the network as a counterexample. Such an execution is defined by discrete actions of the automata in the network and by real-valued delays that describe the time that passes between the discrete actions. Hence, fixing an erroneous system requires insight into both actions and delays that caused the violation on the given counterexample. 

In this paper, we present an approach that facilitates this insight through counterfactual explanations for the observed violation. Previous approaches for explaining real-time violations only consider safety properties~\cite{MariDG21} or only real-time delays without discrete actions~\cite{kolbl2020}, and hence cannot provide a comprehensive insight for violations of unconstrained MITL properties. Like related efforts for discrete systems~\cite{beer2012,coenen2022}, we ground our explanation method in the theory of \emph{actual causality} as formalized by Halpern and Pearl~\cite{halpern2016,HalpernPearl05a,HalpernPearl05b} and identify the actions and delays that are actual causes for the violation of the specification on the observed counterexample. This approach faces several new challenges when confronted with real-time models expressed as networks of timed automata, instead of the previously considered structural equation models~\cite{HalpernPearl05a}, finite-state machines~\cite{coenen2022}, and traces~\cite{beer2012}. 

The first challenge pertains to the concept of \emph{interventions}, which describe how the observed counterexample is modified when hypothetical counterfactual executions are considered during the analysis. While previous results usually consider models where the set of counterfactual scenarios is finite, modifying delays in executions of timed automata gives rise to infinitely many counterfactual scenarios. Our main insight to solve this problem is based on constructing networks of timed automata that model all such counterfactual executions, such that checking a causal hypothesis or even synthesizing a cause from scratch can be realized through model checking of these newly constructed automata. Actual causality in models with infinitely many variables, each potentially having an infinite domain, is only starting to be understood~\cite{halpern2022} and our results suggest that known techniques from timed automata verification are partially transferable to this general theory, e.g., for cause computation.

A second challenge we face in networks of timed automata pertains to the concept of \emph{contingencies}. When two or more potential causes preempt each other, contingencies allow to isolate the true, non-preempted cause from the others. In structural equations models~\cite{halpern2015} and Coenen et al.'s definition for finite-state machines~\cite{coenen2022}, this is realized by extending the system dynamics with resets that set variables back to the value they had in the actual, original scenario. Networks of timed automata have both local variables, i.e., component locations, and global variables such as clocks. We account for this by defining two automata constructions that allow such resets through contingencies on the local level by single components, as well as on the network level for global clock variables.

\begin{figure}
    \centering
    \begin{subfigure}{.49\textwidth} 
    \centering
		\begin{tikzpicture}[draw, semithick, node distance = 8em, every text node part/.style={align=center}, state/.style={draw,circle,minimum size = 4em}]
			\node[state, initial above, initial text = {}] (l0) {$\mathsf{init}$}; 
			\node[state, right of=l0](l1) {$\mathsf{crit}$\\$x \leq 3$}; 
			\path[draw,->]   (l0) edge[bend right = 10] node[below] {$\beta$\\$x := 0$} (l1)
                            (l1) edge[bend right = 10] node[above] {$\beta$\\$x = 3$} (l0)
                            (l0) edge [loop left,looseness=5]  node[left] {$\alpha$\\$x := 0$} (l0);
                            
		\end{tikzpicture}
        \caption{A component automaton of network $\TA_1 \, || \, \TA_2$.}
    \label{fig:motivationa}
    \end{subfigure}
    \begin{subfigure}{.49\textwidth} 
    \centering	
    \begin{tikzpicture}[draw, semithick, node distance = 3em]
    \node[](o1){$\{\mathsf{init_{1,2}}\}$};
    \node[right =0.75 of o1](o2){$\{\mathsf{crit_1}, \mathsf{init_2}\}$};
    \node[right =0.75 of o2](o3){$\{\mathsf{crit_{1,2}}\}$};
\node[below right= 0.5 and -0.5 of o1](u1){$\{\mathsf{init_1}, \mathsf{crit_2}\}$};
\node[right =0.75 of u1](u2){$\bigl(\{\mathsf{init_{1,2}}\}$};
\node[right =0.75 of u2](u3){$\bigl)^\omega$};
\draw[->](o1) edge node[above]{\small$\beta_1$} node[below]{\small$1.0$} (o2)
(o2) edge node[above]{\small$\beta_2$} node[below]{\small$1.0$} (o3)
(u1) edge node[above]{\small$\beta_2$} node[below]{\small$1.0$} (u2)
(u2) edge node[above]{\small$\alpha_1$} node[below]{\small$2.0$} (u3);
\draw[->](o3.south) -- ++ (0,-0.25) -- ++ (-5.5,0) |- node[pos=0.8,above]{\small$\beta_1$} node[pos=0.8,below]{\small$2.0$} (u1.west);
		\end{tikzpicture}
  \vspace{0.5em}
        \caption{An execution of the network $\TA_1 \, || \, \TA_2$.}
    \label{fig:motivationb}
    \end{subfigure}
    \caption{The network and its execution discussed as an illustrative example in Subsection~\ref{subsec:example}.}
    \label{fig:motivation}
\end{figure}

%$\{\mathsf{init_{1,2}}\} \xrightarrow[1.0]{\beta_1} \{\mathsf{crit_1}, \mathsf{init_2}\} \xrightarrow[1.0]{\beta_2}\{\mathsf{crit_{1,2}}\}$}
%$\xrightarrow[2.0]{\beta_1}\{\mathsf{init_1}, \mathsf{crit_2}\}\xrightarrow[1.0]{\beta_2}\bigl(\{\mathsf{init_{1,2}}\} \xrightarrow[2.0]{\alpha_1}\bigl)^\omega$
\subsection{Illustrative Example}\label{subsec:example}

We discuss our approach for causal analysis with the example of a small network of timed automata $\TA_1 \, || \, \TA_2$ consisting of two identical component automata as depicted in Figure~\ref{fig:motivationa}, which we will also use as a running example throughout the paper. The two automata can each switch between the two locations $\mathsf{init}$ and $\mathsf{crit}$, but whenever they enter the location $\mathsf{crit}$ with action $\beta$, they are required to stay there for exactly three time units. This is realized through an initial reset of a global clock variable ($x := 0$) with the first $\beta$ action and a location invariant ($x \leq 3$) in location $\mathsf{crit}$, as well as a clock guard ($x = 3$) on the second $\beta$ action. We want to check the mutual exclusion property $\TLAlways_{[0, \infty)} (\lnot \mathsf{crit_1} \lor \lnot \mathsf{crit_2})$ expressed in MITL, which states that the two automata $\TA_1$ and $\TA_2$ are never both in location $\mathsf{crit}$. It is easy to see that this property is violated, e.g., by the execution depicted in Figure~\ref{fig:motivationb}. This (simplified) execution is an infinite sequence of location labels constructed from delays and discrete actions, where the $\omega$-part is repeated infinitely often. For abbreviation, we place the delay values and actions on the same arrow, which means that the action above the arrow is performed after delaying for as long as specified under the arrow. Both action and location labels refer to a component automaton performing the action and being in a location, respectively, through their index. The execution depicted in Figure~\ref{fig:motivationb} respects the dynamics of the automata, e.g., exactly three time units pass between entering and leaving $\mathsf{crit}$. As we can see, Automaton 2 uses a $\beta$ action less than three time units after Automaton 1, while the latter needs to stay in $\mathsf{crit}$ for exactly three time units.

\begin{table}
    \caption{A contrastive overview of the four root causes on the execution in the illustrative example, inferred using but-for causality and actual causality.}
    \label{tab:intro_table}
\centering
		\def\arraystretch{1.15}
		\setlength\tabcolsep{1mm}
		\begin{booktabs}{X[c]ccc}
			\toprule
			\textbf{Ref.} & \textbf{But-For Causes} & \textbf{Actual Causes} & \textbf{Intuitive Description}\\
			\midrule
                \textbf{1} & $\{(1.0,1,\TA_1) \}$ & $\{(1.0,1,\TA_1) \}$ & The first component did not wait. \\
                \midrule[dotted]
                \textbf{2} & $\{(2.0,1,\TA_2) \}$ & $\{(2.0,1,\TA_2) \}$ & The second component did not wait. \\
                \midrule[dotted]
                \textbf{3} & \makecell{\textbf{a:} $\{(\beta, 1, \TA_1), (\beta, 2, \TA_1)\}$ \\ \textbf{b:} $\{(\beta, 1, \TA_1), (3.0, 2, \TA_1)\}$} & $\{(\beta,1,\TA_1) \}$ & The first component entered $\mathsf{crit}$. \\
                \midrule[dotted]
                \textbf{4} & $\{(\beta,1,\TA_2)\}$ & $\{(\beta,1,\TA_2) \}$ & The second component entered $\mathsf{crit}$. \\
			\bottomrule
		\end{booktabs}
\end{table}

We generate explanations through counterfactual reasoning: For instance, we can infer that one cause of the violation above is that the second component waits only two time units before entering $\mathsf{crit}$ by considering hypothetical executions with alternative delays at this particular point, all else being the same. This relaxed model allows an execution where the second component waits with entering the $\mathsf{crit}$ location until after the first component has already left theirs, such that no violation occurs. Hence, we can infer the \emph{but-for} cause $\{(2.0,1,\TA_2) \}$ which says that the first delay of $2.0$ time units by component $\TA_2$ is a root cause for the violation. We measure delays locally on the component level and hence need to add all global delays between actions of the second component as defined in the execution above. Table~\ref{tab:intro_table} lists this cause (Cause~2) along with the other root causes inferred through \emph{but-for} causal analysis in the second column. Cause~1 expresses that the delay of component $\TA_1$ can similarly be set high enough that no violation occurs. 

Cause~3 shows that the \emph{but-for} counterfactual analysis is not always enough: With this na\"ive criterion we cannot infer that the first $\beta$ action of $\TA_1$ is a cause for the violation of the property on this execution, since changing it alone to, e.g., $\alpha$, does not suffice to avoid the violation. The second $\beta$ then steps in to produce the same effect, which means we are dealing with a \emph{preemption} of potential causes. In the but-for causal analysis, we consequently have to additionally intervene on the preempted causes to obtain executions to avoid the effect. In this case, we can either additionally change the second $\beta$ (Cause~3a), or the second delay (Cause~3b -- this way we can set the entering of $\mathsf{crit}$ to after component $\TA_2$ has already left). These larger causes are not desirable, because they do not only point to the root of the issue.  As a solution in such cases of preemption, Halpern and Pearl~\cite{HalpernPearl05a} suggest \emph{contingencies}, and Coenen et al.~\cite{coenen2022} have recently lifted this to finite-state machines with infinite executions. Inspired by these efforts, we propose a contingency mechanism for networks of timed automata that similarly allows us to infer the first $\beta$ as the true cause in the given scenario (Cause~3). This mechanism extends the network with contingency edges that, e.g., allow the second $\beta$ to move to the same location as in the original execution, i.e., to $\mathsf{init}$. This then produces a witnessing counterfactual run that avoids the effect.

\subsection{Outline and Contributions}

After recalling preliminaries in Section~\ref{sec:prelim}, we develop our definitions of counterfactual causality in networks of timed automata (Section~\ref{sec:causalitydef}). We follow Halpern's approach~\cite{halpern2016} in first defining a notion of \emph{minimal but-for} causality. Counterfactual reasoning is realized through an automaton construction that allows to search for a witnessing intervention in the infinite set of counterfactual runs through model checking. Inspired by Coenen et al.~\cite{coenen2022,coenen2022-1}, we extend but-for causality through a construction of contingency automata, which model contingencies on the local level of components as well as on the network level for global variables such as clocks (Subsection~\ref{subsec:contingencies}), yielding a main building block for our definition of \emph{actual} causality. In Section~\ref{sec:algorithms}, we present algorithms for computing and checking but-for and actual causes. These algorithms exploit a property of both notions of causality that we term cause \emph{monotonicity}, which allows us to reduce the potential causes we need to consider during computation. We have implemented a prototype of this algorithm and report on its experimental evaluation in Section~\ref{sec:experiments}. We show that causes can be computed in reasonable time and help in narrowing down the behavior responsible for an MITL violation. To summarize, we make the following contributions:

\begin{itemize}
    \item We define and study the notions of but-for causality and actual causality in networks of timed automata, for effects described by arbitrary MITL properties;
    \item We propose an algorithm for computing these causes and study its theoretical complexity;
    \item We report the results of a prototype implementation of this algorithm for automated explanations of counterexamples in real-time model checking.
\end{itemize}

\section{Preliminaries}\label{sec:prelim}

We recall background on actual causality, timed automata as models of real-time systems, and MITL as a temporal logic for specifying real-time properties.

\subsection{Actual Causality}\label{subsec:hpcausality}
We recall Halpern's modified version~\cite{halpern2015} of actual causality~\cite{HalpernPearl05a}, which uses \emph{structural equation models} to define the causal dependencies of a system. Formally, a \emph{causal model} is a tuple $M = (\mathcal{S}, \mathcal{F})$ that consists of a \emph{signature} $S = (\mathcal{U}, \mathcal{V}, \mathcal{R})$ and \emph{structural equations} $\mathcal{F} = \{F_X \, | \, X \in \mathcal{V}\}$. The sets $\mathcal{U}$ and $\mathcal{V}$ define \emph{exogenous variables} and \emph{endogenous variables}, respectively. The \emph{range} $\mathcal{R}(Y)$ specifies the possible values of each variable $Y \in \mathcal{Y} = \mathcal{U} \cup \mathcal{V}$. A structural equation $F_X \in \mathcal{F}$ defines the value of an endogenous variable $X \in \mathcal{V}$ as a function $F_X: (\times_{Y \in \mathcal{Y} \setminus \{X\}} \, \mathcal{R}(Y)) \to \mathcal{R}(X)$ of the values of all other variables in $\mathcal{U} \cup \mathcal{V}$, without creating cyclic dependencies in $\mathcal{F}$. Therefore, the structural equations have a unique solution for a given \emph{context} $\vec{u} \in (\times_{U \in \mathcal{U}} \, \mathcal{R}(U)) $, i.e., a valuation for the variables in $\mathcal{U}$. Actual causality then defines whether a value assignment $\vec{X} = \vec{x}$ causes $\varphi$, a conjunction of \emph{primitive events} $Y = y$ for $Y \in \mathcal{V}$, in a given context.

\begin{definition}[Halpern's Version of Actual Causality~\cite{halpern2015}]\label{def:actualcauseHP}
	$\vec{X} = \vec{x}$ is an actual cause of $\varphi$ in $(M, \vec{u})$, if the following three conditions hold: 
	\begin{description}
		\item[AC1.] $(M, \vec{u}) \models \vec{X} = \vec{x}$ and $(M, \vec{u}) \models \varphi$.
		\item[AC2.] There is a \emph{contingency} $\vec{W} \subseteq \mathcal{V}$ with  $(M, \vec{u}) \models \vec{W} = \vec{w}$ and a setting $\vec{x}'$ for the variables in $\vec{X}$ s.t. $(M, \vec{u}) \models [\vec{X} \leftarrow \vec{x}', \vec{W} \leftarrow \vec{w}]\neg\varphi$.
		\item[AC3.] $\vec{X}$ is minimal, i.e., no strict subset of $\vec{X}$ satisfies {\textbf{AC1}} and {\textbf{AC2}}.
	\end{description}
\end{definition}

AC1 simply states that both the cause and the effect have to be satisfied in the given context $\vec{u}$ and causal model $\mathcal{M}$. AC2 appeals to an \emph{intervention} $\vec{X} \leftarrow \vec{x}'$ that overrides the structural equations for all $\vec{X}_i \in \vec{X}$ such that $F_{\vec{X}_i} = \vec{x}'_i$. While the witness $\vec{x}'$ can be chosen arbitrarily, the valuation $\vec{w}$ for the \emph{contingency} variables $\vec{W}$ has to be the same as in the original context. 
The contingency is applied after the intervention, and in this way allows to reset certain variables to their original values, with the aim to infer more precise causes in certain scenarios. Hence, AC2 requires that some intervention together with a contingency avoids the effect, i.e., the resulting solution to the modified structural equations falsifies at least one primitive event in $\varphi$. AC3 ensures that $\vec{X} = \vec{x}$ is a concise description of causal behavior by enforcing minimality. In particular, this ensures that for no variable the valuation in $\vec{x}'$ (AC2) coincides with its original valuation in $\vec{x}$. 

\begin{example}
    We recall a classic example of Suzy and Billy throwing rocks at a bottle~\cite{HalpernPearl05a}. We have the endogenous variables $\mathit{BT},\mathit{ST}$ for Billy and Suzy throwing their rock, respectively. $\mathit{BH},\mathit{SH}$ signify that they hit, and $\mathit{BB}$ encodes that the bottle breaks from a hit. $\mathit{BT}$ and $\mathit{ST}$ directly depend on some nondeterministic exogenous variables, while the other structural equations are $\mathit{BH} = \mathit{BT} \land \lnot \mathit{ST}$, $\mathit{SH} = \mathit{ST}$ and $\mathit{BB} = \mathit{BH} \lor \mathit{SH}$, i.e., Suzy's throw is always faster than Billy's. Hence, in the context where both throw their rock, we have $\mathit{BT} = \mathit{ST} = \mathit{SH} = \mathit{BB} = 1$ and $\mathit{BH} = 0$. The \emph{intervention} $\mathit{ST} = 0$ does not suffice to avoid the effect, because the structural equations still evaluate to $\mathit{BB} = 1$ due to Billy's throw. We say Billy's throw was \emph{preempted}. We can pick the contingency $\mathit{BH} = 0$ from the original evaluation as a \emph{contingency}. This means we set both $\mathit{ST} = 0$ and $\mathit{BH} = 0$, the latter of which ``blocks'' the influence of Billy's throw, and obtain an evaluation where the effect disappears, i.e., with $\mathit{BB} = 0$. Finally, only the event $\mathit{ST} = 1$ is in the cause.
\end{example}

\subsection{Networks of Timed Automata}
We use networks of timed automata~\cite{alur1999} to model real-time systems. We fix a finite set $\AP$ of \emph{atomic propositions} and a finite set of \emph{actions} $\Act$. Given a set of real-valued \emph{clocks} $X$, a \emph{clock constraint} is a conjunctive formula of atomic constraints of the form $x \sim n$ or $x - y\sim n$ with $x,y \in X$, $\sim \, \in \! \{<, \leq, =, \geq, >\}$, and $n \in \NN$. The set of clock constraints over a clock set $X$ is denoted $\CC(X)$. Then, a \emph{timed automaton} is a tuple $\TA = (\Loc, q_0, X, \trel, I,L)$, where $\Loc$ is a finite set of \emph{locations}, $q_0 \in \Loc$ is the \emph{initial location}, $X$ is a finite set of \emph{clocks}, $\trel \, \subseteq (\Loc \times \CC(X) \times \Act \times \CU(X) \times \Loc)$ is the \emph{edge relation}, $I: \Loc \to \CC(X)$ is an \emph{invariant assignment}, and  $L: \Loc\to 2^\AP$ is a \emph{labeling function}. We consider a version of \emph{updatable} timed automata that can reset clocks to constants~\cite{DBLP:journals/tcs/BouyerDFP04}. Hence, the set of \emph{clock updates} $\CU(X)$ is the set of partial functions mapping clocks to natural numbers: $\CU(X) = \{ U : X \rightharpoonup \mathbb{N} \}$. A \emph{clock assignment} for a set of clocks $X$ is a function $u: X \to \RR_{\geq 0}$. $u_0$ denotes the assignment where all clocks are mapped to zero. 
We write $u \models g$ if $u$ satisfies a clock constraint $g \in \CC(C)$, $u + \delta$ for the clock assignment that results from $u$ after $\delta \in \RR_{\geq 0}$ time units have passed, i.e., $(u + \delta)(x) = u(x) + \delta$, and $u \leftarrow U$ for the assignment that updates $u$ in accordance with $U$, i.e., $(u \leftarrow U)(x) = U(x)$ if $x \in \mathit{dom}(U)$ else $(u \leftarrow U)(x) = u(x)$.

\begin{definition}[Semantics of Timed Automata]
    The semantics of a timed automaton $\TA = (\Loc, q_0, X, \trel, I,L)$ is defined by a transition system $\smash{(\Loc \times \mathbb{R}^{\abs{X}}_{\geq 0}),(q_0,u_0),\rightarrow)}$, where $\rightarrow$ contains:
\begin{description}
    \item[\emph{delays:}] $( q, u ) \xrightarrow{\delta} ( q, u + \delta )$ 
		iff $\delta \in \RR_{\geq 0}$ and $(u+\delta') \models I(q)$ for all $0 \leq \delta' \leq \delta$, and
    \item[\emph{actions:}] $( q, u ) \xrightarrow{\alpha} ( q', u \leftarrow U)$
		iff $(q,g,\alpha,U,q') \in \trel$, $ u \models g$, and $(u \leftarrow U) \models I(q')$.
\end{description}
\end{definition}

A \emph{run} $\rho = (q_0,u_0) \xrightarrow{\delta_1}\xrightarrow{\alpha_1} (q_1,u_1) \xrightarrow{\delta_2}\xrightarrow{\alpha_2} \ldots$ of $\TA$ is a sequence of alternating delay and action transitions. The set $\Pi(\TA)$ is the set of all runs of $\TA$. The \emph{trace}  $\pi(\rho) = \langle \delta^\rho_1, \alpha^\rho_1 \rangle \langle \delta^\rho_2, \alpha^\rho_2 \rangle \ldots $ of a run $\rho$ is the sequence of delay and action transitions.  We sometimes denote the elements of some run $\rho$  or trace $\pi$ at index i with $q^\rho_i$, $\delta^\rho_i$ etc. We define the accumulated delay as $\delta(i,j) =  \sum_{k=i,\ldots, j} \delta_k$ and $\delta_0 = 0$. The \emph{signal} $\sigma^\rho$ of the run $\rho$ maps time points to location labels: $\sigma^\rho(t) = \{ a \; | \; \exists i. \; a \in L(q_i) \land \delta(0,i) \leq t < \delta(0,i+1) \}$. The language $\Lang(\TA)$ is the set of all signals with a corresponding run of $\TA$. We use this \emph{left-closed right-open} interpretation of signals due to Maler et al.~\cite{MalerNP06} because of its simplicity. It is straightforward to extend our counterfactual analysis technique to other semantics, e.g., continuous time and point wise~\cite{alur1996}, or even to other logics with linear-time semantics, as long as their model checking problem is decidable. Note that we make use of an intersection operation $\cap$ for timed automata which intersects the \emph{actions}, i.e., the edge label of the automata. You may assume that the operation unifies the labels of the locations, but we apply it such that only one operand automaton has location labels. This means that the result of $\TA_1 \cap \TA_2$ is not (singal-based) language intersection in the classical sense, i.e., we do not have $\Lang(\TA_1 \cap \TA_2) = \Lang(\TA_1) \cap \Lang(\TA_2)$.

\begin{definition}[Network of Timed Automata]
A network of timed automata $\TA_1 \; || \ldots || \; \TA_n $ is constructed through parallel composition. Let $\TA_i = (\Loc^i, l_0^i, X, \trel^i, I^i,L^i)$ for all $1 \leq i \leq n$ with a common set of global clocks $X$. The network $\TA_1 \; || \ldots || \; \TA_n $ is defined by the automaton $\TA = (\Loc, q_0, X, \trel, I,L)$, where the locations are the Cartesian product $\Loc = \Loc^1\times\ldots\times\Loc^n$, with the initial state $q_0^n = (q^1_0,\ldots,q^n_0)$, the invariants are combined as $I(\vec{q}\,) = \bigwedge_{1\leq i \leq n} I^i(q_i)$, and the labels are unified as $L(q) = \bigcup_{1\leq i \leq n} L^i(q_i)$. The edge relation $E$ contains two types:

\begin{description}
    \item[\emph{internal:}] $\big(\vec{q},g,\langle\TA_i,\TA_i,\alpha\rangle,U,\vec{q}\,[q'_i/q_i]\big)$ 
		iff $(q_i,g,\alpha,U,q'_i) \in E^i$, and
    \item[\emph{synchronized:}] $\big(\vec{q},g_i \land g_j,\langle\TA_i,\TA_j,\alpha\rangle,U_i\cup U_j,\vec{q}\,[q'_i/q_i,q'_j/q_j]\big)$
		iff $i \neq j$, $(q_i,g_i,\alpha,U,q'_i) \in E^i$, and $(q_j,g_j,\bar{\alpha},U,q'_j) \in E^j$.
\end{description}
\end{definition}

\noindent Hence, we do explicitly identify the component automata participating in action transitions by constructing tuples containing actions and automata handles. This is for technical convenience in later constructions, and we define a predicate to check whether an automaton participates in an action transition as $\mathit{participates}(\TA_i,\langle\TA_j,\TA_k,\alpha\rangle) := (i = j) \lor  (i = k)$, as well as a partial function for accessing the original action as $\mathit{action}(\TA_i,\langle\TA_j,\TA_k,\alpha\rangle) = \alpha$ iff $i = j$ and $\mathit{action}(\TA_i,\langle\TA_j,\TA_k,\alpha\rangle) = \bar{\alpha}$ iff $i = k$. 

\subsection{Metric Interval Temporal Logic}

We use \textit{Metric Interval Temporal Logic} (MITL)~\cite{alur1996} for defining real-time properties such as system specifications and effects. The syntax of MITL formulas over a set of atomic propositions $\AP$ is defined by $\phi := p \mid \neg \phi \mid \phi \land \phi \mid \phi \, \TLUntil_I \phi$, where $p \in \AP$ and $I$ is a non-singleton interval of the form $[a,b]$, $(a,b]$, $[a,b)$, $(a,b)$, $(a, \infty)$, or $[a, \infty)$ with $a, b \in \NN$ and $a < b$. We also consider the usual derived Boolean and  temporal operators ($\TLEventually_I \, \phi := \top \, \TLUntil_I \, \phi$, $\TLAlways_I \, \phi := \neg \TLEventually_I \, \neg \phi$,  $\phi \, \TLUntil \psi := \phi \, \TLUntil_{[0, \infty)} \psi, \,$ $\, \TLEventually \, \phi := \TLEventually_{[0, \infty)} \, \phi,\,$ and $\, \TLAlways \, \phi := \TLAlways_{[0, \infty)} \, \phi$). 

\noindent The semantics of MITL is defined inductively with respect to a signal $\sigma: \RR_{\geq 0} \to 2^\AP$ and a timepoint $t \in \RR_{\geq 0}$.
\begin{align*}
	&\sigma, t \models p &\text{iff} &&&p \in \sigma(t) \\
	&\sigma, t \models \neg \phi &\text{iff} &&&\sigma, t \not\models \phi \\
	&\sigma, t \models \phi \land \psi  &\text{iff} &&&\sigma, t\models \phi \text{ and } \sigma, t\models \psi\\
	&\sigma, t \models \phi \, \TLUntil_I \, \psi  &\text{iff} &&& \exists t' > t. ~ t' - t \in I, \; \sigma,t' \models \psi \text{ and } \forall t'' \in (t, t'). \, \sigma, t'' \models \phi
\end{align*}
A run $\rho$ \emph{satisfies} an MITL formula $\phi$, iff $\sigma(\rho), 0 \models \phi$. A timed automaton $\TA$ \textit{satisfies} $\phi$, iff all of its runs satisfy $\phi$. We write $\rho \models \phi$ and  $\TA \models \phi$, respectively.

\section{Counterfactual Causality in Real-Time Systems}\label{sec:causalitydef}

In this section, we develop two notions of counterfactual causality in real-time systems. We first define our language for describing causes and how they define interventions on timed traces (Subsection~\ref{subsec:eventsets}). We then start with a simple notion of \emph{but-for} causality in networks of timed automata based on interventions without contingencies (Subsection~\ref{subsec:butforcausality}). Afterward, we outline how to model contingencies in a timed automaton (Subsection~\ref{subsec:contingencies}) and use them to define \emph{actual} causes, in the sense of Halpern and Pearl (cf.~Subsection~\ref{subsec:hpcausality}). 
Note that the proofs of all nontrivial statements of this section are in Appendix~\ref{sec:proofs3}.

\subsection{Interventions on Timed Traces}\label{subsec:eventsets}

We describe actual causes as finite sets of \emph{events}. Events have two distinct types such that they either refer to an action or a delay transition in a given run.

\begin{definition}[Event]\label{def:event} A delay event is a tuple $(\delta, i) \in \RR_{\geq 0} \times \NN$ and an action event is a tuple $(\alpha, i) \in Act \times \NN$. The sets of all delay and action events are denoted as $\mathcal{DE}$ and $\mathcal{AE}$, respectively. The set of all events is $\events = \mathcal{DE} \, \dot\cup \, \mathcal{AE}$. For a trace $\pi$, the set of events on $\pi$ is defined as $\events_\pi = \{(\alpha^\pi_i,i) \; | \; i \in \mathbb{N}_{>0}\}\cup\{(\delta^\pi_i,i) \; | \; i \in \mathbb{N}_{>0}\}$. 
\end{definition} 

When we describe a cause as a set of events, we are mainly interested in the counterfactual runs obtained by modifying the events contained in the cause. In the style of Halpern and Pearl, we call such modifications \emph{interventions}. If the actual run is given in a finite, lasso-shaped form and the cause is a finite set of events, these interventions can be described by a timed automaton that follows the dynamics of the actual trace, except for events that appear in the cause. For these events, the behavior is relaxed to allow arbitrary alternative actions or delays. We call a run $\rho$ \emph{lasso-shaped} if it can be composed of a (possibly empty) prefix and an infinitely occurring loop, i.e., if it is of the form $$\smash{\rho = (q_0,u_0) \ldots   \big( (q_n,u_n) \ldots (q_{p-1},u_{p-1})  \xrightarrow{\delta_p}\xrightarrow{\alpha_p} \big)^\omega} \enspace , $$
where the $\omega$-part is repeated infinitely often. Note that strictly speaking, a lasso-shaped trace as defined here does not exists for all models that violate an MITL property, because clock valuations are not guaranteed to stabilize in some infinitely-repeating loop $u_n \ldots u_{p-1}$. We use the valuations to reset clock values in our contingency construction that will be introduced in Subsection~\ref{subsec:contingencies}. This construction may be generalized by considering clock \emph{regions} or \emph{zones}~\cite{Bengtsson2003} instead of the valuations. This requires defining the resets in the contingency automaton accordingly. 

In this paper, we simplify by assuming the existence of a lasso-shaped run as defined above. In general, we can further assume the clocks to be assigned to values in $\mathbb{Q}$, as timed automata do not distinguish between the real and rational numbers~\cite{AlurD94}. For a lasso-shaped run $\rho$ as described above, we define a function to access the successor index of an action as $\mathit{dst}_\rho : \{1,\ldots,p\} \mapsto \{0,\ldots,p-1\}$ with $\mathit{dst}_\rho(k) = k$ if $k \neq p$ and $\mathit{dst}_\rho(p) = n$ else. We define the length of the run $\rho$ as $|\rho| = p$. The functions $\mathit{dst}_\pi$ and $|\pi|$ are defined analogously for the trace of a lasso-shaped run. We are now ready to define the automaton modelling traces with interventions. 

\begin{figure}[t]
    \centering 
		\begin{mathpar}
			\inferrule*
			{(\delta^\pi_i, i) \notin \cause\\
			(\alpha^\pi_i, i) \notin \cause}
			{\big(i-1, d = \delta^\pi_i, \alpha^\pi_i,d := 0, \mathit{dst}_\pi(i)\big) \in E}
			\and
			\inferrule*
			{(\delta^\pi_i, i) \in \cause\\
			(\alpha^\pi_i, i) \notin \cause}
			{\big(i-1,\top,\alpha^\pi_i,d := 0,\mathit{dst}_\pi(i)\big) \in E}
			\and
			\inferrule*
			{(\delta^\pi_i, i) \notin \cause \!\! \\
			(\alpha^\pi_i, i) \in \cause \!\! \\
			\beta \in Act}
			{\big(i-1,d = \delta^\pi_i,\beta,d := 0,\mathit{dst}_\pi(i)\big) \in E}
			\and
			\inferrule*
			{(\delta^\pi_i, i) \in \cause \!\! \\
			(\alpha^\pi_i, i) \in \cause \!\! \\
			\beta \in Act}
			{\big(i-1,\top,\beta,d := 0,\mathit{dst}_\pi(i)\big) \in E}
		\end{mathpar}
    \caption{Rules defining the edge relation $E$ of the counterfactual trace automaton $\TA^\cause_\pi$.}
    \label{fig:edge_rules}
\end{figure}

\begin{definition}[Counterfactual Trace Automaton]\label{definiton:CFTA_Delay}
	Let $\pi$ be a lasso-shaped trace over the set of actions $\mathit{Act}$ and let $\cause \subseteq \events_\pi$ be a finite set of events. The counterfactual trace automaton of trace $\pi$ for the set of events $\cause$ is defined as  $\TA^\cause_\pi := (\Loc, q_0, X, \trel, I,L)$ with $\Loc := \{0, \dots, |\pi|-1\}$, $q_0 := 0$, $X := \{d\}$. The transition relation $\trel$ is defined by the following rules depicted in Figure~\ref{fig:edge_rules}, we have $L(q) = \emptyset$ for all $q \in Q$, and $$I (q) := \begin{cases} d \leq \delta^\pi_{q + 1}, & \text{if }(\delta^\pi_{q+1}, q+1) \not\in \cause \\
		\top, & \text{otherwise.}\end{cases} $$
\end{definition}

The main idea of the counterfactual automaton $\TA^\cause_\pi$ is to follow the actions and delays of the original run for all events that are not in the event set $\cause$, and allow arbitrary action and delays for events in $\cause$. Hence, $\TA^\cause_\pi$ modifies the \emph{trace} of $\rho$, i.e., the sequence of action and delay events. Subsequently, we will combine a local $\TA^\cause_\pi$ with the dynamics of the original components to obtain full counterfactual runs of a network of timed automata. The interventions on actions and delays are captured by the rules that define the transition relation and are listed in Figure~\ref{fig:edge_rules}, which treat the different combination of events that may or may not be in the cause at a specific index $i$. Crucially, the automaton $\TA^\cause_\pi$ then captures not just a single concrete intervention on the events in $\cause$ with respect to the run $\rho$, such as a modified trace with a specific alternative delay deviating from the actual trace,
\begin{figure}[b]
    \centering
		\begin{tikzpicture}[draw, semithick, node distance = 13em, every text node part/.style={align=center}, state/.style={draw,circle,minimum size = 3em}]
			\node[state, initial, initial text = {}] (l0) {$d \leq 1$}; 
			\node[state, right of=l0](l1) {}; 
   			\node[state, right of=l1](l2) {$d \leq 2$}; 
      
			\path[draw,->]   (l0) edge node[above] {$\{\alpha, \beta\}$, $d = 1$} (l1)
                            (l1) edge node[above] {$\beta$, $\top$} (l2)
                            (l2) edge [loop right]  node[midway, right] {$\alpha$, $d = 2$\\$d := 0$} (l2);

                \path[->]   (l0) edge node[below] {$d := 0$} (l1)
                (l1) edge node[below] {$d := 0$} (l2);
		\end{tikzpicture}
    \caption{Counterfactual trace automaton $\TA^\cause_\pi$.}
    \label{fig:CTA}
\end{figure} 
 but \emph{all} (possibly infinitely many) interventions on the events, i.e., it contains all traces with possibly varying actions and delays at specific indices.

\begin{example}\label{ex:CTA}
    For the trace $\pi = \langle 1.0, \beta \rangle \langle 3.0, \beta \rangle (\langle 2.0, \alpha \rangle)^\omega$ and the set of events $\cause = \{(\beta, 1), (3.0, 2)\}$, we depict the counterfactual trace automaton $\TA^\cause_\pi$ in~Figure~\ref{fig:CTA}.  For the first action and the second delay, arbitrary interventions are allowed, all other action and delay events are enforced to be as in $\pi$. 
\end{example}

\subsection{But-For Causality in Networks of Timed Automata}\label{subsec:butforcausality}

We now use the construction from the previous section to define counterfactual causes for MITL-expressible effects on runs of networks of timed automata. In practice, an effect $\phi$ may be the violation of a specification $\psi$, such that the effect corresponds to the negation of the specification: $\phi \equiv \lnot \psi$.
The main idea of our definition is to isolate the local traces of the component automata, and then construct a counterfactual trace automaton (cf.~Definition~\ref{definiton:CFTA_Delay}) for each component, where the former intervenes on the events in a given cause that refer to the specific component. Afterward, each counterfactual trace automaton is intersected with its corresponding component automaton, and the network of all these intersections describes the counterfactual runs after intervention. To apply interventions locally, we start by defining the local projections of a run in a network of timed automata.

\begin{definition}[Local Projection]
    For a network $\TA^n = \TA_1 \; || \ldots || \; \TA_n$ and one of its runs $\rho$, we denote $\{j_1, \ldots, j_l\} := \{ \, j \in \mathbb{N} \; | \; \mathit{participates}(\TA_i,\alpha^\rho_j)\}$ as the the event points of some component automaton $\TA_i$, whereby we let $j_1 < \ldots < j_l$. Then the local projection $\rho(\TA_i)$ of the component automaton  $\TA_i$ is defined as the trace $\rho(\TA_i) := \langle \delta_1, \alpha_1 \rangle\langle \delta_2, \alpha_2 \rangle\ldots$, in which
    
    \begin{itemize}
        \item $\alpha^{\rho(\TA_i)}_{k} := \mathit{action}(\TA_i, \alpha^\rho_{j_k})$ for all $k = 1, 2, \ldots$, i.e., the identity of the actions is preserved;
        \item $\delta^{\rho(\TA_i)}_k = \Sigma_{x = j_{k-1} +1,\ldots, j_k} \, \delta^\rho_x$ for all $k = 1, 2, \ldots$ and with $j_0 := 1$,  i.e., the delays in the local projection are the cumulative delays between two actions of the automaton in the global run of the network. 
    \end{itemize} 
    
    \noindent Furthermore, we denote with $\mathit{locations}(\rho, \TA_i) :=(q^\rho_{0})_i, (q^\rho_{j_1})_i,(q^\rho_{j_2})_i, \ldots$ the sequence of local locations, i.e., the projection to the i-th component of the network location.
We define the localization function as $\mathit{localize}(\rho,\TA^n) := (\rho(\TA_1), \ldots, \rho(\TA_n))$.
\end{definition}

Note that the local projection as defined here differs fundamentally from local runs as defined for the local time semantics of timed automata~\cite{DBLP:conf/concur/BengtssonJLY98}, as the clocks still advance globally at the same speed. However, by conducting counterfactual interventions on the delays in a local projection of a run, we are able to change the order of transitions, which is not possible by interacting with delays in the global run of the network. It should also be noted that even if the global run $\rho$ is infinite, the local projections may still turn out to be finite because the transitions occurring infinitely often may stem from a subset of the automata.

\begin{example}
    For the run $\rho$ from~Subsection~\ref{subsec:example}, the first local projection $\rho(\TA_1)$ is exactly the trace $\pi$ considered in~Example~\ref{ex:CTA} and $\mathit{locations}(\rho, \TA_1) = \mathsf{init}, \, \mathsf{crit}, \, (\mathsf{init})^\omega$ as the sequence of local locations. The second local projection $\rho(\TA_2)$ is the finite trace $\rho(\TA_2) = \langle2.0, \beta\rangle \langle3.0, \beta\rangle$.
\end{example}

It is worth pointing out that every global run induces well-defined local projections, however, the tuple $\mathit{localize}(\rho,\TA)$ of local traces may have multiple associated global runs. This stems from nondeterminism in the order of actions happening at the same timepoint. In essence, we treat the scheduler's decisions in such a situation as nondeterministic, and allow different resolution of this nondeterminism in counterfactual runs of the network.

\begin{proposition}
    The localization function is not injective: There exists a network $\TA$ and two runs $\rho \neq \rho'$ such that $\mathit{localize}(\rho,\TA) = \mathit{localize}(\rho',\TA)$.
\end{proposition}

Since we want to apply the construction of the counterfactual trace automaton locally to every component, we lift the definition of events from traces to (network) runs, such that the events of a network run are the union of events on local projections of the run.

\begin{definition}[Events of a Network Run]
    Given a run $\rho$ of a network $\TA_1 \; || \ldots || \; \TA_n$, we define the set of associated events as 
    $$\mathcal{E}_{\rho} := \{ \, (e,i,\TA_k) \mid (e,i) \in \mathcal{E}_{\rho(\TA_k)} \enspace \text{\emph{for some component}} \enspace 1 \leq k \leq n \, \} \enspace .$$
     We lift the set of all events to a network $\TA^n$ and define it as $\mathcal{E}(\TA^n) = \{(e,i,\TA_k) \mid (e,i) \in \mathcal{E} \land 1 \leq k \leq n\}$ and say a run $\rho$ satisfies a set of events $\cause \subseteq \mathcal{E}(\TA^n)$, denoted $\rho \models \cause$, if $\cause$ is a subset of the events on $\rho$, i.e., if $\cause \subseteq \events_\rho$. We further define an operator to filter for events of a specific component $k$: $\mathcal{C}|_{k} := \{(e,i) \mid (e,i,\TA_k) \in \mathcal{C} \} $.
\end{definition}

Note that $\mathcal{E}_{\rho(\TA_k)}$ contains both action events as well as \emph{locally projected} delay events. Hence, when we speak about the events on a network run we talk about the actions of the respective component automata (identified through the third position in the event tuple), as well as about the time between these actions of a component automaton (i.e., the cumulative delays between two actions of a component). These events are the atomic building blocks of our counterfactual expalantions. With this at hand we can define our first notion of counterfactual causality based on allowing arbitrary alternatives for all the events appearing in a hypothetical cause. The corresponding notion for structural equation models was termed \emph{but-for} causality by Halpern~\cite{halpern2016}, so we adopt the same name here. Crucially, in our setting with networks of timed automata, the alternative delays and events are realized with respect to the local projections of the network run, such that an alternative delay can change the order of actions emerging in different component automata.

\begin{definition}[But-For Causality in Real-Time Systems]\label{def:bf-cause-RT-delay}
	Let $\TA_1 \; || \ldots || \; \TA_n $ be a network of timed automata and $\rho$ a run of the network.
	A set of events $\cause \subseteq \mathcal{E}(\TA^n)$ is a but-for cause for $\phi$ in $\rho$ of $\TA$, if the following three conditions hold: 
	\begin{description}
		\item[\textbf{SAT}] $\rho \models \cause$ and $\rho \models \phi$, i.e., cause and effect are satisfied by the actual run.
		\item[\CFBF] There is an intervention on the events in $\cause$ s.t.\ the resulting run avoids the effect $\phi$, i.e., we have $$(\TA_1 \inter \TA^{\cause|_1}_{\rho(\TA_1)}) \; || \ldots || \; (\TA_n \inter \TA^{\cause|_n}_{\rho(\TA_n)}) \not\models \phi \enspace .$$
		\item[\textbf{MIN}] $\cause$ is minimal, i.e., no strict subset of $\cause$ satisfies \emph{\textbf{SAT}} and \emph{\CFBF}.
	\end{description}
\end{definition}

\begin{example}\label{ex:butfor}
    Consider again the system and run from Subsection~\ref{subsec:example}, and the cause $\cause = \{(\beta,1,\TA_1),(\beta,2,\TA_1)\}$, i.e., the two $\beta$-actions of the first component (Cause 3a). \textbf{SAT} is satisfied, since the effect $\lnot \TLAlways_{[0, \infty)} (\lnot \mathsf{crit_1} \lor \lnot \mathsf{crit_2})$, i.e., the negation of the MITL specification is satisfied and the local projection of $\TA_1$ is $\rho(\TA_1) = \langle 1.0, \beta \rangle \langle 3.0, \beta \rangle (\langle 2.0, \alpha \rangle)^\omega$. For \CFBF{}, consider that the network run emerging from setting $\TA_1$'s local projection to $\langle1.0, \alpha\rangle \langle3.0, \alpha\rangle (\langle 2.0, \alpha \rangle)^\omega$ does not violate the specification since the first component never enters $\mathsf{crit_1}$. To see that $\cause$ also satisfies \textbf{MIN} consider its two singleton subsets. Setting the alternative $\alpha$ for either of the actions alone does not suffice to avoid the effect due to the temporal ordering of the $\beta$-actions, e.g., when intervening only on the first $\beta$, then the second $\beta$ enters $\mathsf{crit_1}$ while the second component is also in its critical section, hence the effect is still present. Similarly, we can show $\{2.0, 1, \TA_2\}$, i.e., the first delay of the second component (Cause 2), as well as all the other but-for causes from Table~\ref{tab:intro_table} to be but-for causes for $\phi$ in $\rho$.
\end{example}

Besides this intuitive example, we can prove several sanity properties about but-for causality. These properties concern the existence and identity of causes in certain distinctive cases. First up, we show that the existence of a but-for cause is guaranteed as long as a system run avoiding the effect exists.

\begin{proposition}\label{prop:cause_existence}
	Given an effect $\phi$ and a network of timed automata $\TA^n = \TA_1 \; || \ldots || \; \TA_n $, then for every run $\rho$ of the network in which $\phi$ appears, there is a but-for cause for $\phi$ in $\rho$ of \,$\TA^n$, if and only if there exists a run $\rho'$ of the network with $\rho' \not\models \phi$.
\end{proposition}

Next, we consider the case where there is nondeterminism on the actual run, i.e., when there is another run with the same trace, that does no satisfy the effect. In this case, our definition returns the empty set as a unique actual cause.

\begin{proposition}\label{prop:cause-empty}
	Given an effect $\phi$ and a network of timed automata $\TA^n$, $\emptyset$ is the (unique) but-for cause for an effect $\phi$ on a run $\rho$ of $\TA^n$, if and only if there exists a run $\eta$ of $\TA^n$ with $\mathit{localize}(\rho,\TA^n) = \mathit{localize}(\eta,\TA^n)$ and $\eta \not\models \phi$, i.e., a run with the same local traces as the actual run, that does, however, not satisfy the effect.
\end{proposition}

From a philosophical point of view, the empty set is a desirable verdict: It conveys that the smallest change necessary to avoid the effect does not consist of any changes of delay or action events, instead simply an alternative resolution of the underlying nondeterminism of this trace suffices to obtain a witnessing counterfactual run. Also from a practical perspective, it is helpful to know that the empty set gets returned only in this distinguishable scenario.

\subsection{Contingencies in Networks of Timed Automata}\label{subsec:contingencies}

Actual causality employs a \emph{contingency} mechanism to isolate the true cause in the case of preemption. The key idea of contingencies to overcome this preemption is to reset certain propositions in counterfactual executions to their value as it is in the actual world. Coenen et al.~\cite{coenen2022} have outlined how to model contingencies for lasso-shaped traces of a Moore machine. We now describe a construction that applies this idea to networks of timed automata. The central idea is that the state resets resulting from applying a contingency now do not only reset the discrete machine state, but the clock assignment and the location of the timed automaton, i.e., the full underlying state. However, a central issue in networks of timed automata is that clocks are global variables shared by all component automata of the network, while the location is a local attribute of single components. We respect this dichotomy by allowing location contingencies only by actions of the corresponding component automaton and clock contingencies by any action in the global network. 
This is realized by two automata constructions, i.e., a local one applied to all component automata (for resetting locations) and a global one applied to the full network (for resetting clocks). In both cases, we model the behavior as an updatable timed automaton, as we outline in the following.

\begin{definition}[Location Contingency Automaton]\label{definiton:location_contingencies}
	Let $\rho$ be a lasso-shaped run of a network and the timed automaton $\TA = (\Loc, q_0, X, \trel, I,L)$ a component of this network. The location contingency automaton of $\TA$ and $\rho$ is defined as $\TA^{\mathsf{loc}(\rho)} := (\Loc', q_0', X, \trel', I',L')$ with $\Loc' := \Loc \times \{0, \dots,|\mathit{locations}(\rho,\TA)|-1\}$, $q_0' := \langle q_0, 0 \rangle$, $I' (\langle q, i \rangle) :=  I(q)$,  $L'(\langle q, i \rangle) :=  L(q)$, and $\trel'$ is defined as follows, where $\pi = \mathit{localize}(\rho,\TA)$. 
		\begin{mathpar}
			\inferrule*
			{(q,g,\alpha,U,q') \in E\\ \!\!\!\!i = 1, \dots, |\rho|}
			{(\langle q, i-1 \rangle ,g,\alpha,U,\langle q',\mathit{dst}_{\pi}(i)\rangle) \in \trel'}
			\and
            \inferrule*
			{(q,g,\alpha,U,q') \in E\\ \!\!\!\!i = 1, \dots, |\rho|}
			{(\langle q, i-1 \rangle ,g,\alpha,U,\langle q^\pi_{j},\mathit{dst}_{\pi}(i)\rangle) \in \trel'}
		\end{mathpar}
\end{definition}
The location contingency automaton $\TAloc_\rho$ hence consists of copies of the original system, one for each position in the lasso-shaped local projection $\pi$. With an action transition, it moves from one copy into the next, either following the edge $(q,g,\alpha,U,q')$ of the original system (left rule in Definition~\ref{definiton:location_contingencies}) or moving to the same location $q^\pi_{j}$ as present in $\pi$ at the respective position $\mathit{dst}_\pi(i)$ by applying a contingency (right rule in Definition~\ref{definiton:location_contingencies}). Note that after the end of the loop in the lasso-shaped projection, the transitions are redirected to the copy corresponding to the initial position of the loop by the definition of the function $\mathit{dst}_\pi$. The same principle can now also be applied to global variables. In our setting, this only concerns clocks, but the following definition of the clock contingency automaton can be generalized to all global variables such as integers, if these are included in the system model.

\begin{definition}[Clock Contingency Automaton]\label{definiton:clock_contingencies}
	Let $\rho$ be a lasso-shaped run of a timed automaton $\TA = (\Loc, q_0, X, \trel, I,L)$. The clock contingency automaton of $\TA$ and $\rho$ is defined as $\TA^{\mathsf{clk}(\rho)} := (\Loc', q_0', X, \trel', I',L')$ with $\Loc' := \Loc \times \{0, \dots,|\rho|-1\}$, $q_0' := \langle q_0, 0 \rangle$, $I' (\langle q, i \rangle) :=  I(q)$,  $L'(\langle q, i \rangle) :=  L(q)$, and $\trel'$ is defined as follows. 
		\begin{mathpar}
			\inferrule*
			{(q,g,\alpha,U,q') \in E\\ \!\!\!\!i = 1, \dots, |\rho|}
			{(\langle q, i-1 \rangle ,g,\alpha,U,\langle q',\mathit{dst}_\rho(i)\rangle) \in \trel'}
			\and
            \inferrule*
			{(q,g,\alpha,U,q') \in E\\ \!\!\!\!i = 1, \dots, |\rho|}
			{(\langle q, i-1 \rangle ,g,\alpha,u^\rho_j,\langle q',\mathit{dst}_\rho(i)\rangle) \in \trel'}
		\end{mathpar}
\end{definition}

Note that strictly speaking we have defined clock updates to values in $\mathbb{Q}$, instead of $\mathbb{N}$ as considered in classic decidability results. It is, however, straightforward to scale these values to the natural numbers~\cite{AlurD94}.
Clearly, the signals modeled by the contingency automata subsume the ones by the original automaton, because it is possible to simply never invoke a contingency and, hence, always follow the dynamics of the original system.

\begin{proposition}\label{prop:subsumption}
    For all timed automata $\TA$ and runs $\rho$ of $\TA$, we have that the languages of the contingency automata subsume the language of the original automaton: $\mathcal{L}(\TAloc_\rho) \supseteq \mathcal{L}(\TA)$ and $\mathcal{L}(\TAclk_\rho) \supseteq \mathcal{L}(\TA)$.
\end{proposition}

\begin{definition}[Actual Causality in Real-Time Systems]\label{def:act-cause-RT-delay}
	Let $\TA_1 \; || \ldots || \; \TA_n $ be a network of timed automata and $\rho$ a run of the network.
	A set of events $\cause \subseteq \mathcal{E}(\TA^n)$ is an actual cause for $\phi$ in $\rho$ of $\TA$, if the following three conditions hold: 
	\begin{description}
		\item[\textbf{SAT}] $\rho \models \cause$ and $\rho \models \phi$, i.e., cause and effect are satisfied by the actual run.
		\item[\CFAct] There is an intervention on the events in $\cause$ and a location and clock contingency (denoted by $\mathsf{loc}(\rho)$ and $\mathsf{clk}(\rho)$ resp.) s.t. the resulting run avoids the effect $\phi$, i.e., we have $$\big((\TA_1^{\mathsf{loc}(\rho)} \inter \TA^{\cause|_1}_{\rho(\TA_1)}) \; || \ldots || \; (\TA_n^{\mathsf{loc}(\rho)} \inter \TA^{\cause|_n}_{\rho(\TA_n)})\big)^{\mathsf{clk}(\rho)} \not\models \phi \enspace .$$
		\item[\textbf{MIN}] $\cause$ is minimal, i.e., no strict subset of $\cause$ satisfies \emph{\textbf{SAT}} and \CFAct.
	\end{description}
\end{definition}

\begin{example}\label{ex:actual-cause}
Consider the but-for cause $\cause = \{(\beta,1,\TA_1),(\beta,2,\TA_1)\}$ from Example~\ref{ex:butfor}. This $\cause$ is not an \emph{actual} cause because it does not satisfy the \textbf{MIN} condition of Definition~\ref{def:act-cause-RT-delay}: The subset $\cause' = \{(\beta,1,\TA_1)\}$ satisfies \textbf{SAT} and \CFAct. For \CFAct{} we can use contingencies to neutralize the effect of the second $\beta$ in the local projection $\rho(\TA_1) = \langle 1.0, \beta \rangle \langle 3.0, \beta \rangle (\langle 2.0, \alpha \rangle)^\omega$. Since this action moves to $\mathsf{init}$ in the original run (cf.~Subsection~\ref{subsec:example}), it can also move to this location in the contingency automaton $\TA_1^{\mathsf{loc}(\rho)}$. Hence we find an intervention (setting $\langle 1.0, \boldsymbol{\beta} \rangle$ to $\langle 1.0, \boldsymbol{\alpha} \rangle$) and a contingency (setting the target location of $\langle 3.0, \beta \rangle$ to $\mathsf{init}$) that avoid the effect together. A more detailed construction of the contingency automaton is given in Appendix~\ref{app:contingency}. In fact, $\cause' = \{(\beta,1,\TA_1)\}$ is an actual cause (Cause 3) since additionally to \textbf{SAT} and \CFAct{} it also satisfies \textbf{MIN} -- the empty set does not satisfy \CFAct{}. Again, also all the other actual causes from Table \ref{tab:intro_table} can be shown to fulfill our definition. 
\end{example}

\begin{remark}
    Note that as a consequence of Proposition~\ref{prop:subsumption}, the statements regarding the existence and identity of causes as proven in Propositions~\ref{prop:cause_existence} and~\ref{prop:cause-empty} can be lifted to actual causality, but require replacing the original networks in the equivalence statements by the contingency automata construction used in \CFAct{} (cf.~Definition~\ref{def:act-cause-RT-delay}).
\end{remark}

\section{Computing Counterfactual Causes}\label{sec:algorithms} 

In this section, we develop algorithms for computing but-for and actual causes for any MITL effect.  Proofs and further details related to this section can be found in Appendix~\ref{app:alg}. We only explicitly present the algorithm for but-for causes; for actual causes the central model checking query needs to be substituted (cf.~Definition~\ref{def:act-cause-RT-delay}, Lines \ref{algcom:mc} and \ref{algcom:mc1} in  Algorithm~\ref{alg:computing}). 

In principle, the algorithms are based on enumerating all candidate causes. However, we can speed up this process significantly by utilizing what we term the monotonicity of causes.

\begin{lemma}[Cause Monotonicity]\label{lemma:monotonicity}
	For every network of timed automaton $\TA$, run $\rho$, and effect $\phi$, we have that
	\begin{enumerate}
		\item if a set of events $\cause$ fulfills \emph{\textbf{SAT}} also every subset $\cause' \subseteq \cause$ fulfills \emph{\textbf{SAT}}.
		\item   if a set of events $\cause$ fulfills \emph{\CFBF}~(fulfills \emph{\CFAct}) also every superset $\cause' \supseteq \cause$ fulfills \emph{\CFBF}~(fulfills \emph{\CFAct}).
	\end{enumerate}
\end{lemma}

The second monotonicity property enables efficient checking of the \textbf{MIN} condition, as it suffices to check only the subsets with one element less instead of checking all subsets of a potential cause. For the computation of causes on a given run $\rho$, a naive approach could now simply enumerate all elements of $\mathcal{P}(\events_\rho)$, that is, all subsets of the all events $\events_\rho$ on $\rho$, and check whether they form a cause. By further exploiting monotonicity properties, we can find a more efficient enumeration significantly accelerating the computation of causes. The key idea is to enumerate through the powerset $\mathcal{P}(\events_\rho)$ simultaneously from below (starting with the empty cause and then causes of increasing size) and above (starting with the full cause and then causes of decreasing size). This then allows to exclude certain parts of the powerset from the computation in two ways: First, when finding a set of events $\cause$ fulfilling \CFBF, we can exclude all of its supersets as we know that they cannot satisfy \textbf{MIN}. Second, when finding a set of events $\cause$ not fulfilling \CFBF, we can exclude all of its subsets as the monotonicity of \CFBF~implies that $\cause' \subseteq \cause$ will neither fulfill \CFBF. This idea of the simultaneous enumeration of $\mathcal{P}(\events_\rho)$ is implemented in Algorithm~\ref{alg:computing}.

\begin{algorithm}[t]
	\caption{\textsf{Compute But-For Causes}}    \label{alg:computing}
	\KwIn{network $\TA = \TA_1 \; || \ldots || \; \TA_n$, run $\rho$ of $\TA$ satisfying effect $\phi$, i.e.\ $\rho \models \phi$}
	\KwOut{set of all but-for causes for $\phi$ in $\rho$ of $\TA$}
	$Res_s := \{\}, Res_l := \{\}, Power := \mathcal{P}(\events_\rho)$ \\
        \For{$i = 0, 1, 2, \ldots, \frac{|\events_\rho|}{2}$}{
            \For{$\cause \in Power$ with $|\cause| = i$:}{
                \If(\tcp*[f]{cause found?\hspace{-7pt}}){$(\TA_1 \inter \TA^{\cause|_1}_{\rho(\TA_1)}) \; || \ldots || \; (\TA_n \inter \TA^{\cause|_n}_{\rho(\TA_n)}) \not\models \phi$\label{algcom:mc}}{
                    $Res_s := Res_s \cup \cause$ \\ \label{algcom:adding1}
                    $Power := \{ \cause' \in Power \, | \, \cause \not\subseteq \cause'\}$\tcp*{remove all supersets}\label{algcom:remove1}
                }
            }
            \For{$\cause \in Power$ with $|\cause| = \frac{|\events_\rho|}{2} - i$:}{
                \eIf(\tcp*[f]{cause found?\hspace{-7pt}}){$(\TA_1 \inter \TA^{\cause|_1}_{\rho(\TA_1)}) \; || \ldots || \; (\TA_n \inter \TA^{\cause|_n}_{\rho(\TA_n)}) \not\models \phi$\label{algcom:mc1}}{
                    $Res_l := Res_l \cup \cause$ \\ \label{algcom:adding2}
                }
                {$Power := \{ \cause' \in Power \, | \, \cause' \not\subseteq \cause\}$\tcp*{remove all subsets}\label{algcom:remove2}
                }
            }
        }
        \Return $Res_s \cup \{\cause \in Res_l \, | \, \neg \exists \, \cause' \subsetneq \cause. \; \cause' \in Res_s \cup Res_l\}$\tcp*{filter $Res_l$ for \textbf{MIN}} \label{algcom:result}
\end{algorithm}

\begin{theorem}
    Algorithm~\ref{alg:computing} is sound and complete, i.e., it terminates with 
    $$\textsf{Compute But-For Causes}(\TA, \rho, \phi) = \{\,\cause \, | \, \cause \text{ is a but-for cause for }\phi \text{ in } \rho \text{ of } \TA \, \} \enspace ,$$     
    for all networks $\TA = \TA_1 \; || \ldots || \; \TA_n$ and runs $\rho$ of $\TA$ satisfying an effect $\phi$.
\end{theorem}

While it is clear that our algorithm requires to solve several model checking problems for the effect $\phi$, we can show that we cannot do better: Model checking some formula $\phi$ can be encoded as a cause checking problem. Hence, asymptotically, cause checking and computation scale similar to MITL model checking for the formula $\varphi$.

\begin{theorem}
    Checking and computing causes for an effect $\phi$ on the run $\rho$ in a network of timed automata $\TA$ is $\mathsf{EXPSPACE}(\phi)$-complete.
\end{theorem}

 Note that this discussion on the complexity with respect to the size of the effect abstracts away from the, e.g., the length of the counterexample, which contributes polynomially to cause checking and exponentially to cause computation since we need to check all subsets of events. In practice, we have observed that the bidirectional enumeration of the powerset realized in Algorithm~\ref{alg:computing} significantly speeds up the compuation of causes.

\begin{table}[!t]

  	\caption{Experimental results. $\boldsymbol{n}$: number of automata in the network; $\boldsymbol{|\rho |}$: run length; $\boldsymbol{|\events_\rho|}$: number of events on the run; $\boldsymbol{\# \cause}$: number of but-for/actual causes; $\boldsymbol{|\cause|}$: \emph{average} but-for/actual cause size; $\boldsymbol{t}$: runtime for computing but-for/actual causes}
    \label{tab:evaluation}
		\centering
		\def\arraystretch{1.1}
		\setlength\tabcolsep{3mm}
		\begin{booktabs}{cccccccccccc}
			\toprule
			\textbf{Instance} & $\boldsymbol{n}$ & $\boldsymbol{|\rho |}$ & $\boldsymbol{|\events_\rho|}$ & $\boldsymbol{\# \cause_{\mathit{BF}}}$ & $\boldsymbol{\# \cause_{\mathit{Act}}}$ &$\boldsymbol{|\cause_{\mathit{BF}}|}$ & $\boldsymbol{|\cause_{\mathit{Act}}|}$ & $\boldsymbol{t_{\mathit{BF}}}$ & $\boldsymbol{t_{\mathit{Act}}}$ \\
			\midrule
                \textsc{Run.\ Ex.\ } & 2 & \makecell{5 \\ 6} & \makecell{11 \\ 16} & \makecell{6 \\ 10} & \makecell{5 \\ 7} & \makecell{1.83 \\ 3.2} & \makecell{1.2 \\ 2} & \makecell{5.67s \\ 88.2s} & \makecell{5.42s \\ 128.8s} \\
                \midrule[dotted]
                \textsc{Run.\ Ex.\ } & 3 & \makecell{5 \\ 7} & \makecell{11 \\ 16} & \makecell{6 \\ 6} & \makecell{5 \\ 6} & \makecell{1.83 \\ 1.5} & \makecell{1.2 \\ 1.17} & \makecell{5.70s \\ 55.0s} & \makecell{5.53s \\ 78.6s} \\
                \midrule[dotted]
                \textsc{Run.\ Ex.\ } & 4 & 9 & 19 & 8 & 7 & 1.625 & 1.14 & 279.3s & 331.8s \\
                \midrule
			\textsc{Fischer} & 2 & \makecell{4 \\ 7}& \makecell{12 \\ 20}& \makecell{2 \\ 5}& \makecell{2 \\ 5}& \makecell{1 \\ 1.2}& \makecell{1 \\ 1.2} & \makecell{2.37s \\ 273.1s} & \makecell{9.73s \\ 1499s}\\
			\midrule[dotted]
			\textsc{Fischer} & 3 & \makecell{5 \\ 7} &  \makecell{14 \\ 20} &  \makecell{2 \\ 5}&  \makecell{2 \\ 5}&  \makecell{1 \\ 1.2}&  \makecell{1 \\ 1.2} &  \makecell{3.40s \\ 283.6s} &  \makecell{16.9s \\ 1516s} \\
   			\midrule[dotted]
			\textsc{Fischer} & 4 & \makecell{6 \\ 7} &  \makecell{16 \\ 20} &  \makecell{2 \\ 5}&  \makecell{2 \\ 5}&  \makecell{1 \\ 1.2}&  \makecell{1 \\ 1.2} &  \makecell{4.58s \\ 295.3s} &  \makecell{28.8s \\ 1535s} \\
			\bottomrule
		\end{booktabs}
\end{table}

\section{Experimental Evaluation}\label{sec:experiments}
We have implemented a prototype in Python.\footnote{Our prototype and benchmarks are available on GitHub~\cite{tool}.} For model checking networks of timed automata, we use \Uppaal~\cite{Uppaal} and the library \Pyuppaal~\cite{Pyuppaal}. Our tool can check and compute causes for effects in the fragment of MITL that is supported by the \Uppaal{} verification suite.
We conducted experiments on the running example of this paper, as well as on Fischer's protocol, a popular benchmark for real-time model checking. The experiments were run on a MacBook Pro with an Apple M3 Max and 64GB of memory. The results can be found in Table~\ref{tab:evaluation}. For the running example, the tool found exactly the causes depicted in Table \ref{tab:intro_table}; for Fischer's protocol, we report details in Appendix \ref{app:fisher}. The computed causes narrow down the responsible behavior on a given execution, with the average size of the causes between $1$ and $3.2$ on execution with a large number of events ($|\events_\rho|$). Using contingencies does result in smaller causes (cf.~$\boldsymbol{\mathit{Avg.} \, |\cause_{\mathit{BF}}|}$ vs.\ $|\boldsymbol{\cause_{\mathit{Act}}}|$) on the running example. This is not the case for Fischer's protocol, where but-for and actual causes are identical. These findings suggest some directions for optimization, since computing but-for causes is more efficient than computing actual causes. Since the latter are always subsets of the former, it may be sensible to first compute but-for causes and then refine them by taking into account contingencies. Further, the times in Table~\ref{tab:evaluation} refer to the time to compute \emph{all} causes. Hence, the performance in practical applications may be improved by iteratively presenting the user with (but-for or actual) causes that have already been found during the execution.

\section{Related Work}

Providing explanatory insight into \emph{why} a system does not satisfy a specification has been of growing interest in the verification community: Baier et al.~\cite{baier2021} provide a recent and detailed survey. Most works focus on discrete systems and perform error localization in executions~\cite{ball2003,groce2004,wang2006,jose2011} or by identifying responsible components~\cite{reiter1987,gossler2010,gossler2013,wang2013,gossler2020,assmann2021}. There are also several works on program slicing for analyzing  dependencies between different parts of a program \cite{weiser1984,horwitz1988,harman2001}. The concepts of vacuity and coverage can be used to gain causal insight also in the case of a successful verification~\cite{beer1997,ball2008,hoskote1999,chockler2008}. There are several recent works that take a state-based view of responsibility allocation in transition systems~\cite{DBLP:conf/aaai/BaierBK0P24,DBLP:conf/lics/MascleBFJK21}, but they do not consider infinite state systems where such an approach is not directly applicable. There are several works~\cite{gandalf23,coenen2022-1,finkbeiner2023} that use a notion of distance defined by similarity relations in the counterfactual tradition of Lewis~\cite{Lewis1973}. These are more closely related to our work since the minimality criterion in our definitions of but-for and actual causality can be interpreted as a similarity relation~\cite{finkbeiner2023}. Like this paper, a range of works has been inspired by Halpern and Pearl's actual causality for generating explanations~\cite{beer2012,datta2015,gossler2020,coenen2022,leitner2013,RafieioskoueiBonakdarpour24,BeutnerFFS23}. Our contingency automata constructions are particularly inspired by Coenen et al.~\cite{coenen2022,coenen2022-1}. There is a growing interest in counterfactual causality in models with infinitely many variables or infinite domains~\cite{finkbeiner2023,halpern2022}. In the latter work, Halpern and Peters provide an axiomatic account for counterfactual causes in such (structural equation) models, where variables are further allowed to have infinite ranges. Our results suggest that fragments of structural equation models related to networks of timed automata as studied here may be particularly amenable to cause computation. A correspondence between these modeling formalisms has already been pointed out by the same authors~\cite{PetersH21}, albeit to the even more expressive hybrid automata~\cite{DBLP:conf/hybrid/AlurCHH92} that subsume timed automata. For real-time systems, Dierks et al.\ develop an automated abstraction refinement technique~\cite{clarke2000} for timed automata based on considering causal relationships~\cite{dierks2007}. 
Wang et al.\ introduce a framework for the causal analysis of component-based real-time systems~\cite{wang2013}. Kölbl et al.\ follow a similar direction and propose a repairing technique of timed systems focusing on static clock bounds~\cite{kolbl2019}. In a further contribution, they consider the delay values of timed systems to compute causal delay values and ranges in traces violating reachability properties~\cite{kolbl2020}. Mari et al.\ propose an explanation technique for the violation of safety properties in real-time systems \cite{MariDG21}, their approach is based on their corresponding work on explaining discrete systems~\cite{gossler2019}. Hence, in the domain of real-time systems ours is the first technique to consider arbitrary MITL properties, i.e., safety \emph{and} liveness, as effects, together with both actions \emph{and} real-time delays as causes.

\section{Conclusion}
Based on the seminal works of Halpern and Pearl, we have proposed notions of but-for and actual causality for networks of timed automata, which define counterfactual explanations for violations of MITL specifications. Our definitions rely on the idea of counterfactual automata that represent infinitely many possible counterfactual executions. We then leveraged results on real-time model checking for algorithms that check and compute but-for and counterfactual causes, demonstrating with a prototype that our explanations significantly narrow down the root causes in counterexamples of MITL properties. Interesting directions of future work are to study symbolic causes~\cite{leitner2013,coenen2022-1,FinkbeinerFMS24} in real-time system, i.e., to consider real-time properties specified in MITL or an event-based logic as causes~\cite{leitner2013,caltais2019}, and to develop tools for visualizing~\cite{DBLP:journals/tvcg/HorakCMHFMDFD22} counterfactual explanations in networks of timed automata.

\bibliography{bibliography}
\newpage
\appendix

\section{Proofs of Section 3}\label{sec:proofs3}

\setcounter{proposition}{13}
\begin{proposition}
	Given an effect $\phi$ and a network of timed automata $\TA^n = \TA_1 \; || \ldots || \; \TA_n $, then for every run $\rho$ of the network in which $\phi$ appears, there is a but-for cause for $\phi$ in $\rho$ of \,$\TA^n$, if and only if there exists a run $\rho'$ of the network with $\rho' \not\models \phi$. 
\end{proposition}
\begin{proof}
	Let $\rho$ be a run of such a network with $\rho \models \phi$. We show both direction of the equivalence separately.\\
 
    ``$\Rightarrow\!\!"$: Assume there is a but-for cause $\cause$ for $\phi$ in $\rho$ of $\TA$. From \CFBF, we know that there exists a run $\rho' \in \Pi\big((\TA_1 \inter \TA^{\cause|_1}_{\rho(\TA_1)}) \; || \ldots || \; (\TA_n \inter \TA^{\cause|_n}_{\rho(\TA_n)})\big)$ such that $\rho' \not\models \phi$. Since the components of the network are built from (trace) intersections, is easy to see that $\Pi(\TA_i \inter \TA^\cause_{\rho(\TA_i)}) \subseteq \Pi(\TA_i)$ for all components $ 1 \leq i \leq n$. From the semantics of the network based on parallel composition, it follows that $\Pi\big((\TA_1 \inter \TA^{\cause|_1}_{\rho(\TA_1)}) \; || \ldots || \; (\TA_n \inter \TA^{\cause|_n}_{\rho(\TA_n)})\big) \subseteq \Pi(\TA_1 \; || \ldots || \; \TA_n )$, from which this direction of the claim immediately follows.\\
    
	``$\Leftarrow\!\!"$: Let $\rho'$ be a run of the network $\TA_1 \; || \ldots || \; \TA_n $ with $\rho' \not\models \phi$. We show that the set of events $\mathcal{E}_{\rho}$, i.e., the set of \emph{all} events appearing on the path $\rho$, fulfills the \textbf{SAT} and the \CFBF{} condition: From our initial assumption, it follows that $\rho \models \phi$ and from the definition of $\mathcal{E}_{\rho}$ we have $\rho \models \mathcal{E}_{\rho}$, hence the \textbf{SAT} condition is fulfilled. From the definition of the counterfactual trace automaton, it follows that the language $\TA^{\mathcal{E}_{\rho}|_i}_{\rho(\TA_i)}$ of every component $i$ describes all possible traces, i.e., arbitrary orderings of actions, with arbitrary delays, over the alphabet of actions $\Act$. From this we can deduce that the runs of the network under arbitrary interventions are in fact the runs of the original network, i.e., we have $\Pi\big((\TA_1 \inter \TA^{\mathcal{E}_{\rho}|_1}_{\rho(\TA_1)}) \; || \ldots || \; (\TA_n \inter \TA^{\mathcal{E}_{\rho}|_n}_{\rho(\TA_n)})\big) = \Pi(\TA_1 \; || \ldots || \; \TA_n )$. Since by our initial assumption there exists a $\rho' \not\models \phi$ in $\TA_1 \; || \ldots || \; \TA_n $, we can deduce that \CFBF{} is fulfilled.
	Finally, since $\mathcal{E}_{\rho}$ is finite, it either has a minimal subset that satisfies the two criteria and hence witnesses this direction of our claim, or $\mathcal{E}_{\rho}$ itself is the desired witness.
\end{proof}

\begin{proposition}
	Given an effect $\phi$ and a network of timed automata $\TA^n$, $\emptyset$ is the (unique) but-for cause for an effect $\phi$ on a run $\rho$ of $\TA^n$, if and only if there exists a run $\eta$ of $\TA^n$ with $\mathit{localize}(\rho,\TA^n) = \mathit{localize}(\eta,\TA^n)$ and $\eta \not\models \phi$, i.e., a run with the same local traces as the actual run, that does, however, not satisfy the effect.
\end{proposition}
\begin{proof}
	Let a network $\TA^n = \TA_1 \; || \ldots || \; \TA_n$ be given. First up, it is easy to see that whenever $\emptyset$ is a but-for cause, it is unique: No other set $\cause \neq \emptyset$ can satisfy \textbf{MIN}, since $\emptyset \subset \cause$ and $\emptyset$ satisfies \textbf{SAT} and \CFBF{} by assumption. We proceed with proving the equivalence:\\
 
    ``$\Rightarrow"$: Assume that $\emptyset$ is a but-for cause on some run $\rho$, then from \CFBF{} it follows that there exists a run $\rho' \in \Pi\big((\TA_1 \inter \TA^{\emptyset}_{\rho(\TA_1)}) \; || \ldots || \; (\TA_n \inter \TA^{\emptyset}_{\rho(\TA_n)})\big)$ such that $\rho' \not\models \phi$. From the definition of the counterfactual trace automaton $\TA^{\emptyset}_{\rho(\TA_i)}$ it follows that for all components $i$ and for all $\rho_i \in \Pi(\TA_i \inter \TA^{\emptyset}_{\rho(\TA_i)})$ we have that $\pi^{\rho_i} = \rho(\TA_i)$. From the definition of the localization function it then follows that for all $\zeta \in \Pi\big((\TA_1 \inter \TA^{\emptyset}_{\rho(\TA_1)}) \; || \ldots || \; (\TA_n \inter \TA^{\emptyset}_{\rho(\TA_n)})\big)$ we have that $\mathit{localize}(\rho,\TA) = \mathit{localize}(\zeta,\TA)$, so in particular for $\rho' $, which shows this direction of the claim.\\ 
    
    ``$\Leftarrow"$: Assume there is such an $\eta$ with $\mathit{localize}(\rho,\TA) = \mathit{localize}(\eta,\TA)$ and $\eta \not\models \phi$. It is easy to see that $\emptyset$ trivially satisfies \textbf{SAT} and \textbf{MIN}. Hence, we only need to show that $\Pi\big((\TA_1 \inter \TA^{\emptyset}_{\rho(\TA_1)}) \; || \ldots || \; (\TA_n \inter \TA^{\emptyset}_{\rho(\TA_n)})\big)$ includes $\eta$ (and indeed all runs with the same local traces as $\rho$). This follows from the fact that $\Pi(\TA_i \inter \TA^{\emptyset}_{\rho(\TA_i)})$ includes all runs $\eta_i$ that have the same trace as the local projection of $\rho$ with respect to this component, i.e., all $\eta_i = \rho(\TA_i)$, due to the definition of $\TA^{\emptyset}_{\rho(\TA_i)}$ and of trace intersection. By the definition of parallel composition, we can conclude that $\Pi\big((\TA_1 \inter \TA^{\emptyset}_{\rho(\TA_1)}) \; || \ldots || \; (\TA_n \inter \TA^{\emptyset}_{\rho(\TA_n)})\big)$ includes all $\rho'$ with $\mathit{localize}(\rho,\TA) = \mathit{localize}(\rho',\TA)$, hence it also includes $\eta$, which can then serve as a witness for $\emptyset$ satisfying \CFBF{}, which closes this direction of the equivalence.
\end{proof}

\section{Algorithms and Proofs of Section 4}\label{app:alg}
In this section, we give the algorithm for checking causality and detailed proofs of the statements from Section~\ref{sec:algorithms}. 
We start by proving the monotonicity properties.
\setcounter{lemma}{21}
\begin{lemma}[Cause Monotonicity]
	For every network of timed automaton $\TA_1 \; || \ldots || \; TA_n$, run $\rho$, and effect $\phi$, we have that
	\begin{enumerate}
		\item if a set of events $\cause$ fulfills \emph{\textbf{SAT}} also every subset $\cause' \subseteq \cause$ fulfills \emph{\textbf{SAT}}.
		\item   if a set of events $\cause$ fulfills \emph{\CFBF}~(fulfills \emph{\CFAct}) also every superset $\cause' \supseteq \cause$ fulfills \emph{\CFBF}~(fulfills \emph{\CFAct}).
	\end{enumerate}
\end{lemma}
\begin{proof} We show the two statements separately:

	\begin{enumerate}
		\item Follows by the transitivity of set inclusions: If $\cause$ fulfills \textbf{SAT}, we have that $\rho \models \cause \text{ and } \rho \models E$. Hence, $\cause' \subseteq \cause \subseteq \events_\rho$ and therefore $\rho \models \cause'$ such that also $\cause'$ fulfills \textbf{SAT}. 
		\item Let $\cause$ fulfill \CFBF, that is, there is a counterfactual run $\rho'$ of $(\TA_1 \inter \TA^{\cause|_1}_{\rho(\TA_1)}) \; || \ldots || \; (\TA_n \inter \TA^{\cause|_n}_{\rho(\TA_n)})$ with $\rho \not\models \phi$. Now notice that for $\cause' \supseteq \cause$, also the transition relation of each counterfactual trace automaton of $\cause'$ is a superset of the one of $\cause$ such that we also have $\Pi(\TA^{\cause'|_i}_{\rho(\TA_i)}) \supseteq \Pi(\TA^{\cause|_i}_{\rho(\TA_i)})$.
        Therefore, $\rho'$ is also a run of $(\TA_1 \inter \TA^{\cause'|_1}_{\rho(\TA_1)}) \; || \ldots || \; (\TA_n \inter \TA^{\cause'|_n}_{\rho(\TA_n)})$ such that $\cause'$ fulfills \CFBF.  The proof for \CFAct~works analogously for the  run in intersection of the contingency and counterfactual trace automata.  
	\end{enumerate}
\end{proof}

\begin{algorithm}[t]
	\caption{\textbf{Checking But-For Cause}}\label{alg:checking}
	\KwIn{network $\TA = \TA_1 \; || \ldots || \; \TA_n$, run $\rho$, effect $\phi$, set of events $\cause$}
	\KwOut{``Is $\cause$ a but-for cause for $\phi$ in $\rho$ of $\TA$?''}
	\If(\tcp*[f]{checking \textbf{SAT}}){$\rho \not\models \phi$ \textbf{or} $\cause \not\subseteq \events_\rho$}{\label{com:sat-check}
		\Return \False
	}
 	\If(\tcp*[f]{checking \CFBF}){$(\TA_1 \inter \TA^{\cause|_1}_{\rho(\TA_1)}) \; || \ldots || \; (\TA_n \inter \TA^{\cause|_n}_{\rho(\TA_n)}) \models \phi$}{\label{com:cg-check}
		\Return \False
	}
 	\For(\tcp*[f]{checking \textbf{MIN}}){event $e \in \cause$}{
        $C' := \cause \setminus \{e\} $\\
		\If{$(\TA_1 \inter \TA^{\cause'|_1}_{\rho(\TA_1)}) \; || \ldots || \; (\TA_n \inter \TA^{\cause'|_n}_{\rho(\TA_n)}) \not\models \phi$}{\label{min:falseline1}
			\Return \False\label{min:falseline2}
		}
	}
	\Return \True
\end{algorithm}

Algorithm~\ref{alg:checking} decides whether a given set of events forms a but-for cause. It is a straightforward implementation of Definition~\ref{def:bf-cause-RT-delay} of but-for causality under the use of monotonicity for accelerating the verification of the \textbf{MIN} condition. Hence, we do not give a detailed proof of correctness for Algorithm~\ref{alg:checking} and continue directly with cause computation.

\begin{theorem}
    Algorithm~\ref{alg:computing} is sound and complete, i.e., it terminates with  $$\textsf{Compute But-For Causes}(\TA, \rho, \phi) = \{\cause \, | \, \cause \text{ is a but-for cause for }\phi \text{ in } \rho \text{ of } \TA \},$$
    for all networks $\TA = \TA_1 \; || \ldots || \; \TA_n$ and runs $\rho$ of $\TA$ satisfying an effect $\phi$. 
\end{theorem}
\begin{proof}
    We argue for soundness ($\subseteq$) and completeness ($\supseteq$) separately:\\
    
    ``$\supseteq$'': Let $\cause$ be a but-for cause for $\phi$ in $\rho$ of $\TA$, i.e. fulfilling \textbf{SAT}, \CFBF, \textbf{MIN}. We first notice that the algorithm does then not remove $\cause$ from $Power$ (until it may be added to $Res_s)$: $\cause$ is not removed by Line~\ref{algcom:remove1} since the minimality of $\cause$ implies that it has no subset fulfilling \CFBF; and $\cause$ is not removed by Line~\ref{algcom:remove2} since the monotonicity of \CFBF{} implies that it has no superset not fulfilling \CFBF. Now since $\cause$ fulfills \CFBF, if $|\cause| \leq \frac{\events_\rho}{2}$, $\cause$ is added to $Res_s$ in Line~\ref{algcom:adding1}, if $|\cause| > \frac{\events_\rho}{2}$ it is added to $Res_l$ in Line~\ref{algcom:adding2} and is, in addition, not removed in the last line as $\cause$ fulfills the \textbf{MIN} condition. Therefore, $\cause$ is returned by $\textsf{Compute But-For Causes}(\TA, \rho, \phi)$.\\
    
    ``$\subseteq$'': Let $\cause \in \textsf{Compute But-For Causes}(\TA, \rho, \phi)$. As for all set of events considered by the algorithm, we have $\cause \in \mathcal{P}(\events_\rho)$ and, hence, $\cause \subseteq \events_\rho$ such that $\cause$ fulfills \textbf{SAT}. By definition of the algorithm, $\cause$ is only returned as a result when it was added to $Res_s$ or $Res_l$. This, in turn, is only the case, if $(\TA_1 \inter \TA^{\cause|_1}_{\rho(\TA_1)}) \; || \ldots || \; (\TA_n \inter \TA^{\cause|_n}_{\rho(\TA_n)}) \not\models \phi$. Therefore, $\cause$ fulfills \CFBF. 
    Lastly to establish the \textbf{MIN} condition, we have to show that there are no proper subsets of $\cause$ that fulfill \textbf{SAT} and \CFBF. Towards a contradiction, lets assume there are such  subsets and let $\cause' \subsetneq \cause$ be the minimal one. Then, $\cause'$ is but-for cause and by the first inclusion $\cause' \in \textsf{Compute But-For Causes}(\TA, \rho, \phi)$. Now, if $\cause$ was returned by the algorithm since $\cause \in Res_s$, then $|\cause'| < |\cause|$ implies that the algorithm has considered $\cause'$ earlier. From this point, however, $\cause \not\in Power$, a contradiction. If $\cause$ was returned since $\cause \in Res_l$, the filtering in Line~\ref{algcom:result} results in a contradiction. Therefore, $\cause$ is a but-for cause for $\phi$ in $\rho$ of $\TA$.
\end{proof}

\begin{theorem}
    Checking and computing causes for an effect $\phi$ on the run $\rho$ in a network of timed automata $\TA$ is $\mathsf{EXPSPACE}(\phi)$-complete.
\end{theorem}
\begin{proof}
    Analyzing the computational compexity of Algorithms~\ref{alg:computing} and \ref{alg:checking} shows the two problems of cause checking and computations to be solvalbe in $\mathsf{EXPSPACE}(\phi)$. For showing $\mathsf{EXPSPACE}(\phi)$-hardness, we present a reduction from the model checking problem, that is $\mathsf{EXPSPACE}$-complete \cite{alur1996}. We construct for a timed automaton $\TA = (\Loc, q_0, X, \trel, I, L)$ an extended reduction automaton  $\TA_{\textsf{red}} := (\Loc \, \dot\cup \, \{ s_{\textsf{new}}, q_{\textsf{new}} \} \,, \, s_{\textsf{new}}\,, \, X \, \dot\cup \, \{x_{\textsf{new}}\} \,, \, \trel' \,, \, I' \, , \, L')$
over an extended set of actions $ Act \, \dot\cup \, \{\alpha_{\textsf{new}}, \beta_{\textsf{new}} \}$ and labels $AP \, \dot\cup \, \{p_{\textsf{new}}\}$ whereby $s_{\textsf{new}}$ and $q_{\textsf{new}}$ are fresh locations, $\alpha_{\textsf{new}}$ and $\beta_{\textsf{new}}$ are fresh actions, $x_{\textsf{new}}$ is a fresh clock, $p_{\textsf{new}}$ is a fresh atomic proposition, and we have 
\vspace{1em}
	\item $\trel' \, := \, \trel \cup \, \{(s_{\textsf{new}}, \top, \alpha_{\textsf{new}}, \{x_{\textsf{new}} := 1\}, q_{\textsf{new}}) \, , \, (q_{\textsf{new}}, \top, \alpha_{\textsf{new}}, \emptyset, q_{\textsf{new}}) \, , \, (s_{\textsf{new}}, \top, \beta_{\textsf{new}}, \epsilon, q_0) \}$,  
 
\vspace{1em}
$I' (q) := \left\{\begin{array}{ll} I(q), & q \in \Loc, \\
	x_{\textsf{new}} \leq 0, & q = s_{\textsf{new}},\\
	x_{\textsf{new}} \leq 1, & q = q_{\textsf{new}},\end{array}\right.$ \hspace{2em} and \hspace{2em} $L' (q) := \left\{\begin{array}{ll} L(q), & q \in \Loc, \\
	\{\,\}, & q = s_{\textsf{new}},\\
	\{p_{\textsf{new}}\}, & q = q_{\textsf{new}}.\end{array}\right.$ 
\vspace{1em}

That is, $\TA_{\textsf{red}}$ is an extension of $\TA$ that has a new initial location $s_{\textsf{new}}$ from which a direct transition (delay of 0) to either a second new location $q_{\textsf{new}}$ or to the initial state of the original automaton $\TA$ is enforced. This new automaton has a new run, namely $\rho_{\textsf{red}} :=  (s_{\textsf{new}}, u_0) \xrightarrow{0.0}\xrightarrow{\alpha_{\textsf{new}}} (q_{\textsf{new}}, u_0) \xrightarrow{1.0}\xrightarrow{\alpha_{\textsf{new}}})^\omega$ fulfilling the effect $\phi_{\textsf{red}} := \phi \lor \TLEventually \, p_{\textsf{new}}$.

Instances $(\TA, \phi)$ of the model checking problem are now mapped to instances of the cause checking problem via the reduction $r: (\TA, \phi) \mapsto (\TA_{\textsf{red}}, ~ \rho_{\textsf{red}}, ~ \phi_{\textsf{red}}, ~ \cause_{\textsf{red}})$
	with $\cause_{\textsf{red}} := \{((\alpha_{\textsf{new}}, 1, \rho(\TA_{\textsf{red}})\}$.
    Now, we have that $\TA \not\models \phi$	iff $\cause_{\textsf{red}}$ is a cause for $\phi_{\textsf{red}}$ in $\rho_{\textsf{red}}$ of $\TA_{\textsf{red}}$. 
\end{proof}

\section{Contingency Automaton}\label{app:contingency}
In this section, we illustrate the contingency automaton construction from Example~\ref{ex:actual-cause}. For automaton $\TA_1$, run $\rho$ from Subsection~\ref{subsec:example} and its sequence of local locations $\mathsf{loc}(\rho, \TA_1) = \mathsf{init}, \mathsf{crit}, \mathsf{init} $, the location contingency automaton $\TA_1^{\mathsf{loc}(\rho)}$ is depicted in Figure~\ref{fig:contingencyEx}. Following Definiton~\ref{definiton:location_contingencies}, the contingency automaton is constructed in the following way:

\begin{itemize}
    \item we copy the automaton $|\rho(\TA_1)|$ times, to encode the current step in the states (second component of the tuple);
    \item we redirect the transitions from the original automata (black transitions) to their target location in the next copy;
    \item in each step, we add contingency transitions (red transitions), allowing the location to be reset to what it had been in the corresponding step of the original run $\rho$. 
\end{itemize}

We can now find a counterfactual run in $\TA_1^{\mathsf{loc}(\rho)}$ avoiding the critical section by taking the contingency from $(\textsf{init}, 1) \xrightarrow{\beta} (\textsf{init}, 2)$. That is, the location after the second transition is reset to what it had been in the original run, namely to location $\textsf{init}$. The construction of the clock contingency automaton works in a similar way: We allow additional transitions to reset the clocks as they had been in the original run at the respective positions. 

\begin{figure}[h]
    \scalebox{0.85}{
    \centering
		\begin{tikzpicture}[draw, semithick, node distance = 8em, every text node part/.style={align=center}, state/.style={draw,circle,minimum size = 4.5em}]
			\node[state, initial, initial text = {}] (l0) at (0,0) {$(\mathsf{init}, 0)$}; 
			\node[state](l1) at (0,-5) {$(\mathsf{crit}, 0)$\\$x \leq 3$}; 
               \node[state] (l01) at (6,0) {$(\mathsf{init}, 1)$}; 
			\node[state](l11) at (6,-5) {$(\mathsf{crit}, 1)$\\$x \leq 3$}; 
               \node[state] (l02) at (12,0) {$(\mathsf{init}, 2)$}; 
			\node[state](l12) at (12,-5) {$(\mathsf{crit}, 2)$\\$x \leq 3$}; 
			\path[draw,->]   (l0) edge[bend right = 25] node[below] {$\beta$\\$x := 0$\\} (l11)
            (l1) edge[bend left = 25] node[above] {$\beta$\\$x = 3$\\[-4pt]\phantom{test}} (l01)
            (l0) edge  node[midway, above] {$\alpha$\\$x := 0$} (l01)
            (l1) edge  node[midway, below] {$\alpha$} (l11)
            (l11) edge[red] node[below right] {$\alpha$} (l02)
            (l01) edge[bend left = 35, red] node[above] {$\beta$\\$x := 0$} (l02)
            (l01) edge[bend right = 25] node[below] {$\beta$\\$x := 0$\\} (l12)
            (l11) edge[bend left = 25] node[above] {$\beta$\\$x = 3$\\[-4pt]\phantom{test}} (l02)
            (l01) edge  node[midway, above] {$\alpha$\\$x := 0$} (l02)
            (l11) edge  node[midway, below] {$\alpha$} (l12)
            (l1) edge[bend right = 35, red] node[below] {$\beta$\\$x = 3$} (l11)
            (l0) edge[red] node[above right] {$\alpha$\\$x := 0$} (l11) 
            (l02) edge[bend right = 15] node[left] {$\beta$\\$x := 0$} (l12)
            (l12) edge[bend right = 15] node[right] {$\beta$\\$x = 3$} (l02)
            (l02) edge [loop right]  node[midway, right] {$\alpha$\\$x := 0$} (l02)
            (l12) edge [loop right]  node[midway, right] {$\alpha$} (l12)
            (l02) edge [loop above, red]  node[midway, above] {$\beta$\\$x := 0$} (l02)
            (l12) edge[bend right = 65, red] node[right] {$\alpha$} (l02); 
		\end{tikzpicture}
    }
    \caption{The contingency automaton $\TA_1^{\mathsf{loc}(\rho)}$ from Example~\ref{ex:actual-cause}.}
    \label{fig:contingencyEx}
\end{figure}

\section{Experimental Setup and Results for Fischer's Protocol}\label{app:fisher}
\begin{figure}
    \centering
	\vspace{-30pt}
		\begin{tikzpicture}[node distance = 1.2cm]
		\node(l0) at (0,0) [initial, initial text = {}, circle, minimum size = 1cm, draw] {\footnotesize \textsf{init}};
		\node(l1) at (4.5,0) [circle, minimum size = 1cm, draw] {\footnotesize \textsf{req}};
		\node(l2) at (4.5,-4.5) [circle, minimum size = 1cm, draw] {\footnotesize \textsf{wait}};
		\node(l3) at (0,-4.5) [circle, minimum size = 1cm, draw] {\footnotesize \textsf{crit}};

		\draw[->] (l0) to node[midway, above] {} (l1);
		\draw[->] (l1) to node[midway, left] {\begin{tabular}{c} $x_i \leq 2$ \\ $id := i$ \\ $x_i := 0$ \end{tabular}} (l2);
		\draw[->] (l2) to [out = 45, in = -45]  node[midway, right] {\begin{tabular}{c} $id = 0$ \\ $x_i := 0$ \end{tabular}} (l1);
		\draw[->] (l2) to node[midway, above] {$x_i \geq 2 \land id = i$} (l3);
		\draw[->] (l3) to node[midway, left] { $id := 0$} (l0);

		\path (l0) -- node[sloped] (text) {\begin{tabular}{c} $id = 0$ \\ $x_i := 0$ \end{tabular}} (l1);	
		
		\node(di1) at (4.5,0.7) [circle, minimum size = 0.7cm] {$x_i \leq 2$};				
		\end{tikzpicture}
    \caption{A single component automaton $\TA_i$ of Fischer's protocol network $\TA_1 \, || \, \dots \, || \, \TA_n$.}	
        \label{fig:fischer}
\end{figure}

We report on the details in the experimental evaluation for Fischer's protocol and the causes identified by the tool. Fischer's protocol is a popular real-time mutual exclusion protocol, we depict one component $\TA_i$ in Figure \ref{fig:fischer}. We then test the effect $\phi := \lnot \TLAlways_{[0, \infty)} \lnot \mathsf{crit_1}$ on the network $\TA_1 \, || \, \dots \, || \, \TA_n$, i.e.\ that the first component reaches its critical section. 

We state the tested runs and results exemplary for $n = 2$:  
\begin{align*}
    \rho_1 := &\bigl(\{\mathsf{init_{1,2}}\} \xrightarrow[1.0]{\tau_1} \{\mathsf{req_1}, \mathsf{init_2}\} \xrightarrow[1.0]{\tau_1}
\{\mathsf{req_1}, \mathsf{init_2}\} \xrightarrow[4.0]{\tau_1}
\{\mathsf{wait_1}, \mathsf{init_2}\}
\xrightarrow[1.0]{\tau_1}
\{\mathsf{crit_1}, \mathsf{init_2}\}\xrightarrow[2.0]{\tau_1}\bigl)^\omega \enspace \\[4pt]
\rho_2 := &\{\mathsf{init_{1,2}}\} \xrightarrow[1.0]{\tau_2}  \Bigl(\{\mathsf{init_1}, \mathsf{req_2}\}\xrightarrow[1.0]{\tau_1} \{\mathsf{req_1}, \mathsf{req_2}\} \xrightarrow[1.0]{\tau_2}
\{\mathsf{req_1}, \mathsf{wait_2}\} \\
&\hspace{3cm} \xrightarrow[1.0]{\tau_1}
\{\mathsf{wait_1}, \mathsf{wait_2}\} 
\xrightarrow[3.0]{\tau_1}
\{\mathsf{crit_1}, \mathsf{wait_2}\}\xrightarrow[1.0]{\tau_1}\{\mathsf{init_1}, \mathsf{wait_2}\}\xrightarrow[1.0]{\tau_2}\Bigl)^\omega \enspace 
\end{align*}
The detected causes for those two runs are reported in Table \ref{tab:fisher_table}. As in Fischer's protocol only internal actions are used, the detected root causes only consist out of delay actions. Further, since we do not encounter cases of preemption, but-for and actual causes agree on this example. For $n = 3, 4$ we tested runs with the same lasso-part.  

\begin{table}[h]
    \caption{Overview of the root causes found in the experiments for Fischer's protocol.}
    \label{tab:fisher_table}
\centering
		\def\arraystretch{1.1}
		\setlength\tabcolsep{3mm}
		\begin{booktabs}{X[c]cccc}
			\toprule
			\textbf{Ref.} & $\mathbf{\rho_1}$\textbf{: BF Causes} & $\mathbf{\rho_1}$ \textbf{: Actual Causes} & $\mathbf{\rho_2}$\textbf{: BF Causes} & $\mathbf{\rho_2}$ \textbf{: Actual Causes}\\
			\midrule
                \textbf{1} & $\{(1.0,1,\TA_1) \}$ & $\{(1.0,1,\TA_1) \}$ & $\{(1.0,1,\TA_2)\}$ & $\{(1.0,1,\TA_2)\}$ \\
                \midrule[dotted]
                \textbf{2} & $\{(4.0,3,\TA_1) \}$ & $\{(4.0,3,\TA_1) \}$ & $\{(2.0,1,\TA_1) \}$ & $\{(2.0,1,\TA_1) \}$\\
                \midrule[dotted]
                \textbf{3} &  &  & $\{(2.0,2,\TA_2) \}$ & $\{(2.0,2,\TA_2) \}$\\
                \midrule[dotted]
                \textbf{4} &  &  & $\{(2.0,2,\TA_1)\}$ & $\{(2.0,2,\TA_1)\}$ \\
                \midrule[dotted]
                \textbf{5} &  &  & \makecell{$\{(3.0,3,\TA_1),$\\$ (6.0,3,\TA_2)\}$} &  \makecell{$\{(3.0,3,\TA_1),$\\$ (6.0,3,\TA_2)\}$}\\
			\bottomrule
		\end{booktabs}
\end{table}

\end{document}

%% file: definitions_verification.tex
%%%%%%%%%%%%%%%%%%%%%%%%%%%%%%%%%%%%%%%%%%%%%%%%%%%%%%%%%%%%%
%% Common Latex Definitions %%%%%%%%%%%%%%%%%%%%%%%%%%%%%%%%%
%%%%%%%%%%%%%%%%%%%%%%%%%%%%%%%%%%%%%%%%%%%%%%%%%%%%%%%%%%%%%

\renewcommand{\iff}{\text{ iff }}
\newcommand{\brackets}[1]{\left[{#1}\right]}
\newcommand{\angles}[1]{\left\langle{#1}\right\rangle}
\newcommand{\round}[1]{\left({#1}\right)}
\newcommand{\cev}[1]{\reflectbox{\ensuremath{\vec{\reflectbox{\ensuremath{#1}}}}}} % reversed vec
\newcommand{\abs}[1]{\left|#1\right|}
\newcommand{\semantics}[2][]{\left\llbracket#2\ifthenelse{\isempty{#1}}{\right\rrbracket}{\right\rrbracket_{#1}}}
\newcommand{\powerset}[1]{2^#1}
\newcommand{\trel}{E}

% Number Sets
\newcommand{\BB}{\ensuremath{\mathbb{B}}}
\newcommand{\NN}{\ensuremath{\mathbb{N}}}
\newcommand{\ZZ}{\ensuremath{\mathbb{Z}}}
\newcommand{\QQ}{\ensuremath{\mathbb{Q}}}
\newcommand{\RR}{\ensuremath{\mathbb{R}}}
\newcommand{\Lang}{\mathcal{L}}

%BA Definitions

\newcommand{\deptype}[1]{\ensuremath{\Sigma \, L_{#1}: \List (#1). ~ \BB^{|L_{#1}|} \to \BB}}
\newcommand{\deptypeshort}[1]{\ensuremath{\Sigma \, L_{#1}. ~ \BB^{|L_{#1}|} \to \BB}}
\newcommand{\ListMinus}[2]{#1 \, \backslash  \, #2}
\newcommand{\swap}[1]{\ensuremath{\overset{\leftrightarrow}{#1}}}

\newcommand{\AP}{\mathit{AP}}
\newcommand{\Act}{\mathit{Act}}
\newcommand{\Loc}{\mathit{Q}}
\newcommand{\TA}{\mathcal{A}}
\newcommand{\TAloc}{\TA^{\mathsf{loc}}}
\newcommand{\TAclk}{\TA^{\mathsf{clk}}}
\newcommand{\CFBF}{$\textbf{CF}_{\mathsf{BF}}$}
\newcommand{\CFAct}{$\textbf{CF}_{\mathsf{Act}}$}
\newcommand{\CC}{\mathsf{C}}
\newcommand{\CU}{\mathsf{U}}
\newcommand{\causes}{\mathcal{C}}
\newcommand{\cause}{\mathcal{C}}
\newcommand{\tempcauses}{\mathcal{I}_\causes}
\newcommand{\events}{\mathcal{E}}
\newcommand{\traces}{{\small\mathsf{Traces}}}
\newcommand{\paths}{{\small\mathsf{Paths}}}
\newcommand{\true}{{\small\mathsf{true}}}
\newcommand{\false}{{\small\mathsf{false}}}

\newcommand{\inter}{\cap} %{\, \text{\large $\bigcirc$}  \hspace{-11.8pt} \cap}
\newcommand{\TLAlways}{\LTLglobally}
\newcommand{\TLEventually}{\LTLeventually}
\newcommand{\TLUntil}{\LTLuntil}

\newcommand{\Uppaal}{\textsc{Uppaal}}
\newcommand{\Pyuppaal}{\textsc{PyUppaal}}

\newcommand{\MITL}{MITL} %TODO

\newcommand{\KwInSpace}[1]{\KwIn{\textbf{Input:} #1}}
\newcommand{\KwReturn}{\textbf{return}}
\newcommand{\True}{\textbf{true}}
\newcommand{\Kwfreeline}{\phantom{text}\\}
\newcommand{\False}{\textbf{false}}